\documentclass[11pt]{article}
\pdfoutput=1
\usepackage[margin=1in]{geometry}

\DeclareMathAlphabet{\mathbbold}{U}{bbold}{m}{n}
\usepackage{dsfont}
\usepackage{amsmath,amsfonts,amssymb,amsthm}
\usepackage[utf8]{inputenc}
\usepackage{mathtools}
\usepackage[usenames,dvipsnames,svgnames,table]{xcolor}
\usepackage{enumitem}
\usepackage{microtype}
\usepackage{braket} %Provides \Bra{}, \Ket{}, etc.
\usepackage{mathrsfs}
\usepackage{subfigure}
\usepackage{booktabs}
\usepackage{float}
\usepackage{qcircuit}

\usepackage[pagebackref]{hyperref}
\hypersetup{
    pdftitle={}, % title
    pdfauthor={}, % author
    colorlinks=true, % false: boxed links; true: colored links
    linkcolor=blue, % color of internal links
    citecolor=blue, % color of links to bibliography
    urlcolor=blue % color of external links
}

\usepackage{algorithm}
\usepackage{algorithmic}

\usepackage{braket}
\usepackage{makecell}
\usepackage{cancel}

\usepackage{listings}
\usepackage{multirow}% http://ctan.org/pkg/multirow
\usepackage{graphicx}% http://ctan.org/pkg/graphicx
\usepackage{multirow}% http://ctan.org/pkg/multirow
\usepackage{booktabs}% http://ctan.org/pkg/booktabs
\usepackage{makecell} %Allows line breaks in table cells
\usepackage{enumitem}
\usepackage[capitalise,nameinlink]{cleveref}
\crefname{thm}{Theorem}{Theorems}
\usepackage{microtype}

\renewcommand{\backref}[1]{}

\renewcommand{\backrefalt}[4]{%
\ifcase #1 %
\or
[p.\ #2]%
\else
[pp.\ #2]%
\fi}

\makeatletter
\newcommand{\para}{%
  \@startsection{paragraph}{4}%
  {\z@}{2ex \@plus 3.3ex \@minus .2ex}{-1em}%
  {\normalfont\normalsize\bfseries}%
}
\makeatother

\newcommand\blfootnote[1]{%
  \begingroup
  \renewcommand\thefootnote{}\footnote{#1}%
  \addtocounter{footnote}{-1}%
  \endgroup
}

\newtheorem{fact}{Fact}
\newtheorem{defn}{Definition}
\newtheorem{lem}{Lemma}
\newtheorem{thm}{Theorem}
\newtheorem{prop}{Proposition}

\newcommand{\lemref}[1]{\hyperref[lem:#1]{Lemma~\ref*{lem:#1}}}
\newcommand{\propref}[1]{\hyperref[prop:#1]{Proposition~\ref*{prop:#1}}}
\newcommand{\algoref}[1]{\hyperref[algo:#1]{Algorithm~\ref*{algo:#1}}}
\newcommand{\As}{\mathcal{A}}
\newcommand{\Ss}{\mathcal{S}}
\newcommand{\var}{\mathrm{Var}}
\newcommand{\expect}{\mathbb{E}}
\newcommand{\pr}{\mathrm{Pr}}
\newcommand{\norm}[1]{\|#1\|}

\newcommand{\abs}[1]{\left|#1\right|}
\newcommand{\trans}{^\textsf{T}}

\DeclareMathOperator{\zero}{\textbf{0}}
\DeclareMathOperator{\one}{\textbf{1}}
\DeclareMathOperator{\nandgate}{\textsf{NAND}}
\DeclareMathOperator{\notgate}{\textsf{NOT}}

\DeclarePairedDelimiter{\ceil}{\lceil}{\rceil}

\DeclareMathOperator{\argmax}{argmax}

\DeclareMathOperator{\A}{\mathcal{A}}
\DeclareMathOperator{\C}{\mathcal{C}}
\DeclareMathOperator{\Q}{\mathcal{Q}}
\DeclareMathOperator{\sigmalow}{\sigma_{\text{low}}}
\DeclareMathOperator{\sigmahigh}{\sigma_{\text{high}}}
\DeclareMathOperator{\Oracle}{\mathcal{G}}

\newcommand{\T}[1]{\mathcal{T}^{#1}}
\newcommand{\Tind}{\mathcal{T}}

\newcommand{\EstI}[3]{\mathsf{qEst1}_{#1}\boldsymbol{(}{#2}, \, {#3}\boldsymbol{)}}
\newcommand{\EstII}[3]{\mathsf{qEst2}_{#1}\boldsymbol{(}{#2}, \, {#3}\boldsymbol{)}}
\newcommand{\Order}[1]{O\boldsymbol{(}{#1}\boldsymbol{)}}
\newcommand{\qargmax}{\textsf{\textup{qArgmax}}}
\newcommand{\SolveMdpI}{\textup{\textsf{SolveMdp1}}}
\newcommand{\SolveMdpII}{\textup{\textsf{SolveMdp2}}}
\newcommand{\AI}{\textup{\textsf{\textup{qEst1}}}}
\newcommand{\AII}{\textup{\textsf{\textup{qEst2}}}}
\newcommand{\qEst}{\textup{\textsf{\textup{qEst}}}}
\newcommand{\qEsti}{\textup{\textsf{\textup{qEst}}\{i\}}}
\newcommand{\qEstI}{\textup{\textsf{\textup{qEst1}}}}
\newcommand{\qEstII}{\textup{\textsf{\textup{qEst2}}}}
\newcommand{\Or}{\textup{\textsf{\textup{OR}}}}
\newcommand{\Id}{\textup{\textsf{Identity}}}
\newcommand{\eps}{\epsilon}

\begin{document}

\title{\bfseries Quantum Algorithms for Reinforcement Learning \\ with a Generative Model
}

\author{
\hspace{5em}Daochen Wang\footnote{University of Maryland. \texttt{wdaochen@gmail.com}} 
\and
Aarthi Sundaram\footnote{Microsoft Quantum. \texttt{aarthi.sundaram@microsoft.com}}
\and
Robin Kothari\footnote{Microsoft Quantum. \texttt{robin.kothari@microsoft.com}}
\hspace{5em}
\and 
Ashish Kapoor\footnote{Microsoft. \texttt{akapoor@microsoft.com}}
\and
Martin Roetteler\footnote{Microsoft Quantum. \texttt{martinro@microsoft.com}}
}
\date{}

\maketitle

%=============================================================================
\begin{abstract}
  Reinforcement learning studies how an agent should interact with an environment to maximize its cumulative reward. 
  A standard way to study this question abstractly is to ask how many samples an agent needs from the environment to learn an optimal policy for a $\gamma$-discounted Markov decision process (MDP). For such an MDP, we design quantum algorithms that approximate an optimal policy ($\pi^*$), the optimal value function ($v^*$), and the optimal  $Q$-function ($q^*$), assuming the algorithms can access samples from the environment in quantum superposition. 
  This assumption is justified whenever there exists a simulator for the environment; for example, if the environment is a video game or some other program. 
  Our quantum algorithms, inspired by value iteration, achieve quadratic speedups over the best-possible classical sample complexities in the approximation accuracy ($\epsilon$) and two main parameters of the MDP: the effective time horizon ($\frac{1}{1-\gamma}$) and the size of the action space ($A$). 
  Moreover, we show that our quantum algorithm for computing $q^*$ is optimal by proving a matching quantum lower bound. \blfootnote{\newline Note: a conference version of this work appears as \cite{ConferenceVersion_Rl_2021} (ICML 2021).}
\end{abstract}

%=============================================================================
\section{Introduction}
\label{sec:intro}

Markov Decision Processes (MDPs) are a fundamental mathematical abstraction in reinforcement learning, used to model problems where an agent should take actions in an environment to maximize its cumulative reward. The framework has been successfully applied to problems in healthcare, robotics, engineering, gaming, natural language processing, finance, and so on~\cite{Berteskas_DynamicControl_2000, Bertsekas_Abstract_2013,Szepesvari_AlgorithmsRl_2010, SB_ReinforcementLearning_2018,AgarwalJiangKakadeSun_Reinforcement_2021}. 

Quantum computers are a model of computation based on the laws of quantum mechanics that promise substantially faster algorithms for certain tasks like search and factoring~\cite{Grover_Search_1996, Shor_Factoring_1997}. Recent experiments have achieved key milestones~\cite{Google_QuantumSupremacy_2019}, bringing forward the tantalizing prospect of using quantum computers for real-world impact in the not-so-distant future.

In this paper, we construct quantum algorithms that more efficiently solve the main problems associated with MDPs: approximating an optimal policy, the optimal value function, and the optimal Q-value function. Our algorithms rely on the assumption that we have quantum access to the environment, which we will justify.

We intend this introduction to be accessible to those unfamiliar with quantum computing, and we have delayed technical discussions of quantum algorithms to \Cref{sec:preliminaries}.

\subsection{Problem Setup}

We study an \emph{infinite-horizon discounted MDP}, $M$, with a finite set, $\Ss$, of \emph{states}, where at each state an agent can choose to take an action from a finite set, $\As$, of \emph{actions}. 
Upon taking an action $a\in \As$ at state $s\in \Ss$, the agent receives \emph{reward}\footnote{We use square brackets to index into vectors and functions.} $r[s,a]\in [0,1]$ and transitions to a state $s'\in \Ss$ with some probability $p(s'|s,a)$. 
The last parameter needed to specify $M$ is the \emph{discount factor} $\gamma \in [0,1)$, which discounts the reward the agent receives at later time steps $t$ by a factor of $\gamma^t$. 
Hence $M$ is conveniently summarized by a $5$-tuple, $M=(\Ss,\As,p,r,\gamma)$. 
For convenience, we define $S = |\Ss|$ and $A = |\As|$, the cardinalities of $\Ss$ and $\As$ respectively, and $\Gamma \coloneqq (1-\gamma)^{-1}$ for the \emph{effective time horizon} of the MDP. 

Given such an MDP, the agent's goal is to choose actions to maximize its expected sum of $\gamma$-discounted rewards over infinitely many time steps. 
Following standard practice, we assume the agent has full knowledge of $\Ss$, $\As$, $r$, and $\gamma$, but not $p$ at the outset. 
A primary objective is to compute a deterministic \emph{policy} $\pi:\Ss \to \As$ for the agent that specifies the action $a = \pi(s)$ it should take at $s\in \Ss$ to best achieve its goal with high probability. 

For a given policy $\pi:\Ss\rightarrow \As$, the \emph{value-function} (or simply \emph{value}) of $\pi$, $v^{\pi}:\Ss\rightarrow[0,\Gamma]$, and the \emph{Q-function} of $\pi$, $q^{\pi}: \Ss\times\As \rightarrow [0,\Gamma]$, are defined by
\begin{equation}\label{eqn:vq}
\begin{aligned}
    v^{\pi}[s] &= \mathbb{E} \left[\sum_{t=0}^\infty \gamma^t r[s_t, a_t] \ \middle\vert \ 
    s_0 = s, \ \ \forall i \, \geq \, 0 \, : \, a_i = \pi[s_i] \right], \ \textrm{and}\\
    q^{\pi}[s,a] &=  \mathbb{E} \left[\sum_{t=0}^\infty \gamma^t r[s_t, a_t]\ \middle\vert \ 
    s_0 = s, \ a_0 = a, \ \ \forall i \,\geq \,1 \,:\, a_i = \pi[s_i] \right],
\end{aligned}
\end{equation}
where the expectations are over the probabilistic state transitions, i.e., for all $i \geq 0$, $s_{i+1}$ is sampled from the distribution $p(\cdot|s_i,a_i)$. 
Note that the maximum value that the sums in~\cref{eqn:vq} can take is $\Gamma$, because the reward function is at most $1$, and hence $v^{\pi}[s]$ and $q^{\pi}[s,a]$ are in $[0,\Gamma]$. 
It is known that any such MDP admits an optimal policy $\pi^*:\Ss\rightarrow\As$, in the strong sense that $v^{\pi^*}[s] \geq v^{\pi}[s]$  and $q^{\pi^*}[s,a] \geq q^{\pi}[s,a]$ for all $\pi \in \Pi$, $s\in \Ss$, and $a \in \As$, where $\Pi$ is the space of all policies (which could even contain randomized and non-stationary policies\footnote{In a randomized policy, the action taken at a given $s \in \Ss$ may be probabilistic. A stationary policy is one where the action taken depends only on the current state $s$.})~\cite{AgarwalJiangKakadeSun_Reinforcement_2021}. 
It is common to define $v^*\coloneqq v^{\pi^*}$ and $q^*\coloneqq q^{\pi^*}$.

We can now state our main computational tasks precisely. 
Using $\norm{\cdot}$ for the infinity norm of a vector, for a given MDP $M$, $\epsilon\in (0,\Gamma)$, and $\delta\in(0,1)$, our goal is to compute a policy $\hat{\pi}$ for $M$ such that with probability at least $1-\delta$, it satisfies $\norm{v^*- v^{\hat{\pi}}}\leq \epsilon$.%
\footnote{It is common in the field of MDPs to have $\eps$ denote the additive approximation error of a number in $[0,\Gamma]$, which makes the valid range of $\eps$ be $(0,\Gamma)$. A more natural normalization may be to divide the $q$ and $v$ functions by $\Gamma$ to have $\eps\in(0,1)$, but this changes what the sample complexity expressions look like, making it harder to visually compare our bounds with prior work.} 
In addition, we are interested in the related tasks of computing approximations $\hat{v}$ (resp.~$\hat{q}$) to $v^*$ (resp.~$q^*$) such that $\norm{v^* - \hat{v}}\leq \epsilon$ (resp.~$\norm{q^* - \hat{q}}\leq \epsilon$) with probability at least $1-\delta$.

The goal of this paper is to design algorithms that perform the above computational tasks using as few resources as possible. The resource use of an algorithm is normally quantified by either its time complexity or by the number of samples it draws from the unknown distribution $p(s'|s,a)$. 
In our paper, we study the latter and assume the \emph{generative model} of sampling, as studied by \cite{KearnsSingh_PhasedQlearning_1999,KearnsMansourNg_SparseSampling_2002,Kakade_Thesis_2003}, where we can choose an \emph{arbitrary} $(s,a)\in \Ss\times \As$ and ask a simulator to draw samples $s'\sim p(\cdot|s,a)$. 
Our goal then translates to minimizing the number of uses of the simulator. 
The generative model makes particular sense when the environment is a computer program, in which case the simulator is the program itself.

We now let quantum computing enter the picture. If the simulator is itself a computer program and we have its source code, then we can produce a Boolean circuit $G$ (with size roughly the same as the time complexity of the program) that acts as the simulator, i.e., draws samples from the distribution $p(\cdot|s,a)$. 
We can use the following basic fact in quantum computation to efficiently convert $G$ to a quantum circuit $\Oracle$ (see \cite{Ben73} or \cite[Sec.~1.4.1]{NielsenChuang_QuantumComputation_2000}).

\begin{fact}
Any classical circuit $G$ with $N$ logic gates can be converted to a quantum circuit consisting of $O(N)$ logic gates that can compute on any quantum superposition of inputs; moreover, the conversion is efficient and based on simple conversion rules at the logic gate level. 
\end{fact}

We refer to $\Oracle$ as the (quantum) oracle or simulator and the ability to query it as the (quantum) generative model. The precise behavior of $\Oracle$ is formally defined in \cref{sec:quantum-prelim}.

Under this setup, our goal is to design quantum algorithms approximating $q^*$, $\pi^*$, and $v^*$ that use the quantum simulator $\Oracle$ as few times as possible. We refer to the number of calls a quantum algorithm makes to $\Oracle$ as its (quantum) query or sample complexity. It is fair to compare the quantum sample complexity with the classical sample complexity because, as we have discussed above, $\Oracle$ and $G$ have similar costs at the elementary gate-level. 

Our paper constructs quantum algorithms having significantly less sample complexity than the best-possible classical algorithms. Moreover, we show that our quantum algorithms are either optimal, or optimal assuming $\Gamma$ or $A$ is constant, for certain ranges of $\epsilon$.

\subsection{Main Results}

\begin{table*}[tb]
    \centering
    \renewcommand{\arraystretch}{1.75}
    \begin{tabular}{c@{\qquad}cc@{\quad}ccc}
        \toprule
        \multirow{2}{*}{\thead{Goal: Output \\an $\epsilon$-accurate \\estimate of}} & Classical sample complexity &  \multicolumn{2}{c}{Quantum sample complexity} \\
        \cmidrule{2-4}
        & Upper and lower bound &  Upper bound & Lower bound 
        \\
        \midrule
        $q^*$ & $\frac{SA\Gamma^3}{\epsilon^2}$ & $\frac{SA\Gamma^{1.5}}{\epsilon}$ \hspace{0.5em} [\Cref{thm:complexity_1}] & $\frac{SA\Gamma^{1.5}}{\epsilon}$ \hspace{1em} [\Cref{thm:lower_bound}] 
        \\[2ex]
        %\hline
        \multirow{2}{*}{$v^*$, $\pi^*$} & \multirow{2}{*}{$\frac{SA\Gamma^3}{\epsilon^2}$} &  
        {$\frac{SA\Gamma^{1.5}}{\epsilon}$ \hspace{0.5em} [\Cref{thm:complexity_1}]} & 
        \multirow{2}{*}{$\frac{S\sqrt{A}\Gamma^{1.5}}{\epsilon}$ \hspace{0.5em} [\Cref{thm:lower_bound}] } \\
        & & $\frac{S\sqrt{A}\Gamma^{3}}{\epsilon}$ \hspace{0.5em} [\Cref{thm:complexity_2}] & 
        \\
        \bottomrule
    \end{tabular}
    \caption{Quantum computing allows for speedups in terms of the parameters  $\epsilon$, $\Gamma \coloneqq (1-\gamma)^{-1}$, and $A$, but not $S$. All bounds are for maximum failure probability $\delta$ constant. All upper bounds are $\widetilde{O}(\cdot)$, with unrestricted $\epsilon$ except when [\Cref{thm:complexity_1}] appears, in which case we assume $\epsilon\in O(1/\sqrt{\Gamma})$. All lower bounds are $\Omega(\cdot)$ and hold for any $\epsilon\in (0,\Gamma/4)$. The classical upper bounds are shown in \cite{Li_TightUpper_2020} for all $\epsilon$; the classical lower bounds are shown in \cite{AzarMunosKappen_MdpGenerative_2012} for $q^*,v^*$ and \cite{Sidford_NearOptimal_2018} for $\pi^*$. We also reprove all three classical lower bounds in \Cref{thm:lower_bound}.}
    \label{tab:results2}
\end{table*}

\Cref{tab:results2} summarizes our main results. The classical sample complexities have only recently been completely characterized for all three quantities~\cite{Li_TightUpper_2020} for the full range of $\epsilon \in (0,\Gamma]$. As the table shows, for computing $q^*$, we construct a quantum algorithm that offers a quadratic speedup in terms of $\Gamma$ and $\epsilon$ if $\epsilon = O(1/\sqrt{\Gamma})$. For computing $v^*$ and $\pi^*$, we construct a second quantum algorithm that offers an additional quadratic speedup in terms of $A$ at the expense of $\Gamma$. Moreover, we prove quantum lower bounds for computing all three quantities. Our lower bounds show that our $q^*$ algorithm is optimal, that we have optimal algorithms for $v^*$ and $\pi^*$ provided one of $\Gamma$ or $A$ is constant, but that there may still be a faster quantum algorithm for $v^*$ and $\pi^*$. We remark that we also reprove the \emph{classical} lower bounds in a qualitatively stronger way than existing bounds as explained at the end of the next section.

We remark that the time complexities of our quantum algorithms are the same as their sample complexities up to log factors assuming that the classical generative model can be called in constant time and that we have access to quantum random access memory (QRAM)~\cite{GiovannettiLloydMaccone_Qram_2008}. This is because the classical algorithm of \cite{Sidford_NearOptimal_2018} that we quantize satisfies this property and the quantum subroutines we use to quantize it also satisfy this property.

\subsection{Technical Overview}\label{sec:technical_overview}

We now give an overview of the techniques we used in our two quantum algorithms, \SolveMdpI\ and \SolveMdpII. \SolveMdpI\ and \SolveMdpII\  correspond to the complexities next to [\Cref{thm:complexity_1}] and [\Cref{thm:complexity_2}] in \cref{tab:results2} respectively. Our two quantum algorithms are essentially the product of infusing quantum subroutines into a modern variant of (approximate) value iteration by \cite{Sidford_NearOptimal_2018}. We first discuss the quantum subroutines: quantum mean estimation~\cite{Brassard_AmplitudeEstimation_2000,Montanaro_MonteCarlo_2015} and quantum maximum finding~\cite{Durr_MinFinding_1996}.

\para{Quantum subroutines.} Quantum mean estimation consists of two similar quantum algorithms \qEstI\ and \qEstII\, that we also refer to collectively as \qEst. Here, \qEst\ can compute the mean $\mathbb{E}[X]$ of a random variable $X$, suitably encoded quantumly, quadratically more efficiently than what is possible classically. \qEstI\ roughly corresponds to a quadratically more sample-efficient Hoeffding's inequality while \qEstII\ roughly corresponds to a quadratically more sample-efficient Chebyshev's (or Bernstein's) inequality. That is, getting additive error $\epsilon$ using these quantum algorithms takes quadratically fewer samples than what those classical inequalities imply. For example, Chebyshev's inequality states that $O(\var[X]/\epsilon^2)$ samples is required; \qEstII\ roughly states that only $O(\sqrt{\var[X]}/\epsilon)$ quantum samples is required. Using quantum mean estimation in both \SolveMdpI\ and \SolveMdpII\ yields the speedups in $\Gamma$ and $\epsilon$.

Quantum maximum finding, denoted \qargmax, is an algorithm that can find the maximum of a list of $n$ numbers, again suitably encoded quantumly, using only $O(\sqrt{n})$ queries to that list. $\qargmax$ is used in \SolveMdpII\ and is the source of its speedup in $A$.
\begin{algorithm*}[htbp]
\caption{$\SolveMdpI(M,\epsilon, \delta)$}
\label{algo:SolveMdp1}
\begin{algorithmic}[1]
    \STATE {\bfseries Input:} MDP $M = (\Ss,\As, p, r, \gamma)$, maximum error $\epsilon \in (0,\sqrt{\Gamma}]$, and maximum failure probability $\delta\in(0,1)$.
    
    \STATE {\bfseries Output:} $\hat{v} \coloneqq v_{K,L} \in \mathbb{R}^S$, $\hat{\pi} \coloneqq \pi_{K,L} \in \As^S$, and $\hat{q} \coloneqq q_{K,L} \in \mathbb{R}^{SA}$.

    \STATE {\bfseries Initialize:} $K  \gets \ceil{\log_{2}(\Gamma/\epsilon)}$, $L \gets \Gamma \ceil{\ln (4\Gamma/\epsilon)} + 1$, $f \gets \delta/4KLSA$, $b \gets 1$, $c \gets 0.01$

    \STATE {\bfseries Initialize:} $v_{1,0} \gets \zero$, $\pi_{1,0}\gets \text{arbitrary}$, $q_{1,0}\gets \zero$

    \FOR{$k\in [K]$}

    \STATE $\epsilon_k \gets \Gamma/2^{k}$

    \STATE $\forall(s,a) \in \Ss\times \As:$

    \STATE \quad $y_k[s,a] \gets \max \{\EstI{f}{(Pv_{k,0}^2)[s,a]}{b} - (\EstI{f}{(Pv_{k,0})[s,a]}{(1-\gamma)b})^2 , 0 \}$

    \STATE \quad $x_k[s,a] \gets \EstII{f}{(Pv_{k,0})[s,a]}{c(1-\gamma)^{1.5}\epsilon \sqrt{y_k[s,a]+b}} - c(1-\gamma)^{1.5}\epsilon\sqrt{y_k[s,a]+b}$
    
    \FOR{$l\in [L]$}

    \STATE $\forall s\in \Ss$: 
    \algorithmicif \ $v(q_{k,l-1})[s] \geq v_{k,l-1}[s]$ \
    \algorithmicthen \ $v_{k,l}[s]\gets v(q_{k,l-1})[s]$, $\pi_{k,l}[s]\gets \pi(q_{k, l-1})[s]$

    \STATE 
    \algorithmicelse \  $v_{k,l}[s]\gets v_{k,l-1}[s]$, $\pi_{k,l}[s]\gets \pi_{k,l-1}[s]$ \ {\bfseries end if}

    \STATE $\forall (s,a)\in \Ss\times \As:
    \Delta_{k,l}[s,a] \gets \EstI{f}{(P(v_{k,l}-v_{k,0}))[s,a]}{c(1-\gamma)\epsilon_k} - c(1-\gamma)\epsilon_k$
    
    \STATE $q_{k,l} \gets \max \{r+\gamma(x_{k}+\Delta_{k,l}), \zero\}$
    \ENDFOR

    \STATE $v_{k+1,0}\gets v_{k,L}$, $\pi_{k+1,0} \gets \pi_{k,L}$, $q_{k+1,0}\gets q_{k,L}$
    \ENDFOR
\end{algorithmic}
\end{algorithm*}

\begin{algorithm*}[htbp]
    \caption{$\SolveMdpII(M, \epsilon, \delta)$}
    \label{algo:SolveMdp2}
\begin{algorithmic}[1]
    \STATE {\bfseries Input:} MDP $M = (\Ss,\As, p, r, \gamma)$, maximum error $\epsilon \in (0,\Gamma]$, and maximum failure probability $\delta\in(0,1)$.
    
    \STATE {\bfseries Output:} $\hat{v} \coloneqq v_{L} \in \mathbb{R}^S$ and $\hat{\pi} \coloneqq \pi_{L} \in \As^S$.

    \STATE {\bfseries Initialize:} $L \gets \Gamma \ceil{\log (4\Gamma/\epsilon)} + 1$, $f \gets \delta/4c_{\max}LSA^{1.5}\log(1/\delta)$

    \STATE {\bfseries Initialize:} $v_0 \gets \zero$, $\pi_0 \gets \text{arbitrary}$, $\forall s\in \Ss: q_{0,s}\gets \zero\in \mathbb{R}^A$

    \FOR{$l\in [L]$}
    
    \STATE
    $\forall s\in \Ss: a^*[s] \gets  \qargmax_f\{q_{l-1,s}[a] \ : \ a\in \A \}$

    \STATE 
    $\forall s\in \Ss: \tilde{\pi}_{l}[s] \gets a^*[s]$, $\tilde{v}_{l}[s] \gets q_{l-1,s}[a^*[s]]$

    \STATE $\forall s\in \Ss$: 
    \algorithmicif \ $\tilde{v}_{l}[s] \geq v_{l-1}[s]$ \
    \algorithmicthen \ $v_{l}[s]\gets \tilde{v}_{l}[s]$, $\pi_{l}[s]\gets \tilde{\pi}_{l}[s]$

    \STATE 
    \algorithmicelse \  $v_{l}[s]\gets v_{l-1}[s]$, $\pi_{l}[s]\gets \pi_{l-1}[s]$ \ {\bfseries end if}

    \STATE $\forall s \in \Ss:
    \text{create quantum oracle encoding, } U_{z_{l,s}} \text{, of } z_{l,s}\in \mathbb{R}^A \text{ defined by}$

    $\hspace{37.5pt} z_{l,s}[a] \gets \EstI{f}{(Pv_{l})[s,a]}{(1-\gamma)\epsilon/4} - (1-\gamma)\epsilon/4$

    \STATE $\forall s \in \Ss: 
    \text{create quantum oracle encoding, } U_{q_{l,s}} \text{, of } q_{l,s}\in \mathbb{R}^A \text{ defined by}$

    $\hspace{37.5pt} q_{l,s}[a] \gets \max\{r[s,a]+\gamma z_{l,s}[a],0\}$
    \ENDFOR
\end{algorithmic}
\end{algorithm*}

\para{Quantum version of standard value iteration.} We will be discussing how the above subroutines can be used in the modern variant of value iteration by \cite{Sidford_NearOptimal_2018}. To warm up, consider how they can be applied to standard value iteration~\cite{KearnsSingh_PhasedQlearning_1999} to compute $v^*$. In standard value iteration, we start with $v_0$ set to the zero vector in $\mathbb{R}^\Ss$ and repeatedly update it by the Bellman recursion $v_i\gets \Tind(v_{i-1})$ where the Bellman operator $\Tind: \mathbb{R}^\Ss \to \mathbb{R}^\Ss$ is defined by
\begin{equation}
    \Tind(v_{i+1})[s] \coloneqq \max_{a}\left\{r[s,a] + \gamma \mathbb{E}\left[v_{i}[s'] \ \middle| \ s'\sim p(\cdot|s,a)\right]\right\},
\end{equation}
for all $s\in \Ss$. For convenience, we denote the mean $\mathbb{E}[v_{i}[s'] \ | \ s'\sim p(\cdot|s,a)]$ by $\mu_i$. Hypothetically, if this mean were computed exactly at each iteration, then this is a contraction map with contraction factor $\gamma$ and fixed point $v^*$. So after $t$ iterations, the error in the current iterate has dropped by a factor of $\gamma^t$. Neglecting log factors, after about $O(\Gamma)$ iterations, our iterate is $\eps$-close to $v^*$. In reality, we cannot compute $\mu_i$ exactly. But if we only require our final answer to be correct to error $\epsilon$, then it is reasonable to assume that estimating $\mu_i$ for each $i$ to error $O(\epsilon/\Gamma)$ suffices. If we make the reasonable assumption $\norm{v_i}\leq \norm{v^*}\leq \Gamma$ ($v_i$ is converging to $v^*$ after all) then classically doing this estimation at \emph{each} iteration uses $O(SA \Gamma^2/(\epsilon/\Gamma)^2) = O(SA \Gamma^4/\epsilon^2)$ samples classically by the Hoeffding bound. The factor $SA$ comes from the fact that an estimation is done for each $(s,a)\in \Ss\times \As$. Therefore, the overall classical sample complexity is of order $O(SA \Gamma^5/\epsilon^2)$. Though the preceding argument is non-rigorous, it does give the right answer (up to log-factors)~\cite{Sidford_SWWY_2021}. 

How would our quantum subroutines speed up standard value iteration? By using quantum mean estimation, we can quadratically suppress the sample complexity at each iteration and for each $(s,a)\in \Ss\times\As$, meaning that the quantum sample complexity at each iteration becomes $O(SA \sqrt{\Gamma^2/(\epsilon/\Gamma)^2}) = O(SA\Gamma^2/\epsilon)$. Accounting for the $\Gamma$ iterations, gives an overall quantum sample complexity of $O(SA \Gamma^3/\epsilon)$. In fact, observing that the Bellman recursion involves taking the maximum over the set of actions, we can use quantum maximum finding to reduce the complexity down further, to $O(S\sqrt{A} \Gamma^3/\epsilon^2)$, which matches the performance of \SolveMdpII\ for $v^*$. However, an
$\epsilon$-optimal value function leads only to an $(2\gamma \Gamma\epsilon)$-optimal greedy policy~\cite{SinghYee_PolicyLoss_1994,Bertsekas_Abstract_2013}.

\para{Quantum version of modern value iteration.} To obtain an $\epsilon$-optimal policy, \SolveMdpI\ and \SolveMdpII\ directly employ the so-called monotonicity technique of \cite{Sidford_NearOptimal_2018} which we observe does not interfere with our use of the two quantum subroutines. The monotonicity technique comprises the if-then-else statement and the subtractions in the lines involving \qEst. Note that the subtracted terms always equal the preceding estimation error which enforces one-sided error. Overall, the monotonicity technique ensures that the value function at each iteration is at most the value function \emph{of the policy} at that iteration (which in turn is at most $v^*$). Hence we avoid the problem of an $\epsilon$-optimal $\hat{v}$ not giving an $\epsilon$-optimal $\hat{\pi}$.

We can get better dependence in $\Gamma$ by leveraging two other techniques introduced in \cite{Sidford_NearOptimal_2018,Sidford_SWWY_2021,Wainwright_VarianceReduced_2019}: ``variance reduction'' and ``total variance''. We incorporate these techniques in \SolveMdpI\ at the cost of re-inflating the $A$ dependence back to linear. The reason we no longer get $\sqrt{A}$ is because applying $\qargmax$ is incompatible with the variance reduction technique.

Variance reduction essentially splits standard value iteration into $K \coloneqq \ceil{\log_2(\Gamma/\epsilon)}$ epochs where in each epoch we halve the error. Epochs in \SolveMdpI\ are indexed by $k$. At the $l$-th iteration of epoch $k$, we need to estimate $\mathbb{E}[v_{k,l}[s']]$, where $v_{k,l}$ is the current value function. The mean can be rewritten as
\begin{equation}\label{eq:variance_reduction}
    \mathbb{E}[v_{k,l}[s']] =  \mathbb{E}[(v_{k,l}-v_{k,0})[s']] + \mathbb{E}[v_{k,0}],
\end{equation}
where $v_{k,0}$ is the value function at the start of the epoch. There are $SA$ of these equations, one corresponding to each $(s,a)\in \Ss \times \As$ such that $s'\sim p(\cdot|s,a)$. We estimate the mean on the left-hand-side (LHS) by the sum of estimates of means on the right-hand-side (RHS). Since $\norm{v_{k,l}-v_{k,0}}$ decreases rapidly with $k$, because $v_{k,l}$ and $v_{k,0}$ rapidly approach $v^*$, we ignore the first term on the RHS in our overview. We remark that its estimation cost affects the $\epsilon$ range for which \SolveMdpI\ is optimal. Consider the second term, $\mathbb{E}[v_{k,0}]$. This again needs to be estimated to error $\epsilon/\Gamma$ which classically costs $O(SA\Gamma^2/(\epsilon/\Gamma)^2) = O(SA\Gamma^4/\epsilon^2)$ by the same argument before. Quantumly, this costs $O(SA\Gamma^2/\epsilon)$, again as before. Now, the key point is that we only need to estimate $\mathbb{E}[v_{k,0}]$ \emph{once per epoch} and reuse its value throughout the epoch. As there are only logarithmically many epochs, the overall cost becomes $\widetilde{O}(SA\Gamma^4/\epsilon^2)$ classically and $\widetilde{O}(SA\Gamma^2/\epsilon)$ quantumly.

The total variance technique is more subtle. It is based on the observation that the actual error accumulation from iteration to iteration is much less than what is implicit above. To be clear, in the above, we set the error in mean estimation at each iteration to be $\epsilon/\Gamma$ so that over $\Gamma$ iterations, the accumulated error is $\epsilon$. However, the error at each iteration $i$ can actually be set larger, to $\epsilon\sqrt{\var[v_{i}[s']]}/\Gamma^{1.5}$ (which could be as large as $\epsilon/\sqrt{\Gamma}$), and it can still be shown that the overall accumulated error is $\epsilon$ using properties of the standard deviation. More specifically, let us write $\sigma_i = \sqrt{\var[v_{i}[s']]}$. Then, the cumulative standard deviation, $\sum_{i=1}^\Gamma\sigma_i$, is closely related to an expression for which we can non-trivially upper bound by $\sqrt{2}\,\Gamma^{1.5}$~(\cref{thm:azar}). Classically, it is straightforward to estimate $\mu_i \, (\coloneqq \mathbb{E}[v_i[s']])$ to an error of $\epsilon \sigma_i/\Gamma^{1.5}$, without needing to know the $\sigma_i$s. This can be done using about $O((\epsilon/\Gamma^{1.5})^{-2}) = O(\Gamma^3/\epsilon^2)$ samples for each state--action pair as guaranteed by Chebyshev's (or Bernstein's) inequality. Combined with variance reduction, that is, applying the above technique to estimate the $\mathbb{E}[v_{k,0}]$ from before, we see that this yields an overall classical sample complexity of $\widetilde{O}(SA\Gamma^3/\epsilon^2)$. This is one main result of \cite{Sidford_NearOptimal_2018}. Due to the first term on the RHS of \cref{eq:variance_reduction}, which we glossed over, this result only holds for $\epsilon=O(1)$.

Trying to do a quantum version of the total variance technique poses a significant technical challenge for the following reason. The version of quantum mean estimation that should have corresponded to a more efficient Chebyshev's inequality, namely \qEstII, is deficient compared to its classical counterpart in two ways. The first is that \qEstII\ cannot estimate $\mu_i$ to an error proportional to $\sigma_i$ without knowing $\sigma_i$ a priori. To remedy this, we first estimate $\sigma_i$ using \qEstI\ to some additive error $b>0$. Denote the estimate by $\hat{\sigma}_i$. Then we can use \qEstII\ to estimate $\mu_i$ to error proportional to $\sigmalow \coloneqq \hat{\sigma}_i-b \,(\leq \sigma_i)$ which maintains correctness. Unfortunately, this approach does not work due to the second deficiency of \qEstII. In fact, \qEstII\ also requires an upper bound $C$ on $\sigma_i$ to function and uses $O(C/\epsilon)$ samples to guarantee additive error $\epsilon$. For large $C$, the sample complexity can be highly redundant with respect to the error guaranteed. This problem is directly relevant for us if we try to use $\sigmahigh \coloneqq \hat{\sigma}_i+b$ as $C$. Then, the complexity becomes proportional to $C/\sigmalow = (\hat{\sigma_i}+b)/(\hat{\sigma}_i-b)$, which can be arbitrarily large depending on the value of $\hat{\sigma}_i$ that we cannot control. To remedy this second problem, we in fact estimate $\mu_i$ to error proportional to $\sigmahigh$, so that $C/\sigmahigh = 1$ becomes constant. Of course, this no longer maintains correctness as $\sigmahigh$ is larger than $\sigma_i$. However, we can bound $\sigmahigh \leq \sigma_i + 2b$. We then find, by performing a full correctness analysis, that the extra error of $2b$ can be sufficiently suppressed if we set $b$ and the parameter $c$ on Line 3 of \SolveMdpI\ to be small enough constants. Doing so only increases the overall complexity by a constant factor. Setting $b$ constant also ensures that the complexity of estimating $\sigma_i$ to error $b$ by \qEstI\ is within our budget. With the technical challenges resolved, we see that the complexity of \SolveMdpI\ is $\widetilde{O}(SA(\epsilon/\Gamma^{1.5})^{-1}) = \widetilde{O}(SA\Gamma^{1.5}/\epsilon)$. Again, due to the first term on the RHS of \cref{eq:variance_reduction}, this only holds for $\epsilon = O(1/\sqrt{\Gamma})$. The $\epsilon$ range is smaller than before, which was $\epsilon=O(1)$, because there is relatively less quantum speedup for estimating that first term. (Note added: subsequently to the conference version of this work appearing~\cite{ConferenceVersion_Rl_2021}, Hamoudi~\cite[Theorem 13]{Hamoudi_SubGaussian_2021} removed the deficiencies of quantum mean estimation, as described in this paragraph, in general.)

In summary, we have described \SolveMdpI, which uses \qEst\ to ``quantize'' all three techniques in \cite{Sidford_NearOptimal_2018}: monotonicity, variance reduction, and total variance. Quantizing the first two is not difficult but quantizing the last one offers a technical challenge. We believe that our solution to that challenge could find uses in quantizing other classical algorithms as well. We have also described \SolveMdpII, which offers a quadratic speedup in $A$ using $\qargmax$. But because $\qargmax$ conflicts with the variance reduction and total variance techniques, \SolveMdpII\ no longer has optimal $\Gamma$ dependence.

\para{Lower bound techniques.} Lastly, we discuss how we prove our lower bounds. Standard techniques for proving lower bounds on the number of uses of a quantum oracle generally work with Boolean oracles. In our case, we instead have an oracle $\Oracle$ that outputs a particular quantum state for a given state--action pair which can also be invoked in superposition over state--action pairs. To enable the use of standard lower bound techniques from quantum query complexity, we reduce the problems of computing certain Boolean functions $f$ to our problems of computing $q^*$, $v^*$, and $\pi^*$ by instantiating our oracle $\Oracle$ using standard Boolean oracles. For example, consider a quantum oracle $\Oracle_{\text{coin}}$ that produces a state which represents a quantum sample of a coin toss with probability $p$ of getting heads. $\Oracle_{\text{coin}}$  can be instantiated by a Boolean oracle encoding a $n$-bit string (for large $n$) which has $p$ fraction of its bits equal to $1$. The reduction then allows us to translate known lower bounds on computing $f$ using a Boolean oracle to lower bounds on computing $q^*$, $v^*$, and $\pi^*$ using oracle $\Oracle$.

This approach has some unexpected benefits. Because we reduce to standard problems in query complexity, our proof is very modular. It allows us to also show optimal \emph{classical} lower bounds by simply invoking the best classical lower bounds for the Boolean functions $f$ mentioned above. Moreover, we qualitatively improve on known classical lower bounds. The known lower bound of \cite{AzarMunosKappen_MdpGenerative_2012} shows that for any $S$, $A$, there exists a hard MDP which has a number of state--action pairs equal to $SA$. However, it is not the case that their constructed MDP has $S$ states and $A$ actions, just that the total number of state--action pairs is $SA$. Their constructed MDP actually has $O(SA)$ states, but most states only have $O(1)$ actions, so the total number of state--action pairs is $SA$. In contrast, our hard MDP instance genuinely has $S$ states and $A$ actions.

\subsection{Related Work}

As we have discussed, our quantum algorithms can be viewed as ``quantizations'' of the classical algorithms and techniques in \cite{Sidford_NearOptimal_2018, Sidford_SWWY_2021,Wainwright_VarianceReduced_2019} which represent the latest development of classical \emph{model-free} MDP solvers, which also recently include \cite{Wang_PrimalDual_2017,Wang_Randomized_2020,Sidford_Mirror_2020} among others, that started with \cite{KearnsSingh_PhasedQlearning_1999}. \cite{Sidford_NearOptimal_2018} give algorithms with complexity $\widetilde{O}(SA\Gamma^3/\epsilon^2)$ when $\epsilon=O(1)$ for approximating all three of $q^*$, $v^*$, and $\pi^*$. On the model-free side, there has been even more recent progress culminating in the work of \cite{Li_TightUpper_2020} which achieves $\widetilde{O}(SA\Gamma^3/\epsilon^2)$ for the full range of $\epsilon \in (0,\Gamma]$. That this bound is tight (up to log-factors) is established by \cite{AzarMunosKappen_MdpGenerative_2012} which is closely related to our work. Indeed, to prove our lower bounds, we use an instance inspired by \cite{AzarMunosKappen_MdpGenerative_2012}. However, our proof by reduction and composition theorems is technically quite different from theirs and extends their lower bound to apply to arbitrary $S$ and $A$. Arguably, model-based MDP solvers~\cite{AzarMunosKappen_MdpGenerative_2012,Agarwal_MinimaxOptimal_2020,Li_TightUpper_2020} have seen more successes than their model-based counterparts that we quantized. However, quantizing these techniques appears more difficult. As a first step, one might ask if the quantum sample complexity of learning a probability distribution supported on $n$ points to error $\epsilon$ in $\ell_1$-norm can be $O(n/\epsilon)$, which represents a quadratic speedup over classical in terms of $\epsilon$. Recently, this question has been answered affirmatively~\cite{VanApeldoorn_MultiDimension_2021,Cornelissen_MultiMonte_2021} which immediately implies an $\widetilde{O}(S^2A\Gamma^2/\epsilon)$ model-based quantum algorithm for $q^*$ due to \cite[Proposition~2.1]{AgarwalJiangKakadeSun_Reinforcement_2021}. However, this complexity is highly suboptimal and it remains to be seen whether we could eventually obtain an optimal model-based quantum algorithm for any one of $q^*$, $v^*$, or $\pi^*$.

On the quantum side, the broader subject of reinforcement learning ``remains relatively unaddressed by the quantum community''~\cite{Jerbi_Rl_2020}. The relatively few works on the subject include \cite{Dong_QuantumRl_2008,Dunjko_MachineLearning_2016,Paparo_ActiveLearning_2014, Dunjko_Advances_2017,
Jerbi_Rl_2020}. However, these works are incomparable to ours as they focus either on problem formulation or lack rigorous results. None give rigorous complexity bounds on computing $\pi^*$, $v^*$, and $q^*$. Some of these works do mention the possibility of quadratic speedups by using quantum maximum finding~\cite{Dunjko_MachineLearning_2016}. However, they do not consider how this technique would work within an integrated algorithm. As we have mentioned, our work shows that to achieve optimal $\Gamma$-dependence overall, we may have to forgo the use of quantum maximum finding. We note that in the multi-armed bandits setting, where $S=1$, an instance-optimal quadratic quantum speedup is shown in \cite{WangYouLiChilds_Bandits_2021}.  In synergy with our work, \cite{Dunjko_Advances_2017} proposes methods to instantiate the quantum generative model in real physical environments as opposed to being given a classical simulator. If their methods can be realized, our work will have wider applicability.

\section{Preliminaries}
\label{sec:preliminaries}

\subsection{Notation}

For a positive integer $n$, we write $[n]$ for the set $\{1,\dots, n\}$. We use upper case letters for matrices and lower case letters for vectors. For vectors only, we use square bracket notation $v[i]$ to mean entry $i$ of vector $v$. Vectors $v$ appearing in this work often have indices $i=(i_1,i_2)$ described by two coordinates in which case we write $v[i_1,i_2]$ to mean $v[(i_1,i_2)]$. As a function $v: X \to Y$ can be identified with the corresponding vector $v\in Y^X$, we also use square bracket notation to index into functions. For any two real vectors $u,v$ of the same dimension, we write $\max \{u,v\}$ to mean the element-wise max of $u$ and $v$ and $u\leq v$ to mean the inequality holds element-wise. We write bold $\one$ (resp.~$\zero$) for a vector of all $1$s (resp~$0$s) with dimension determined by context. A scalar $x\in \mathbb{R}$ appearing alone in an equation involving vectors is to be interpreted as $x\cdot \one$. For a function $f: A \to B$ and vector $v$ with entries in $A$, we write $f(v)$ for the vector with entries in $B$ resulting from applying $f$ to $v$ element-wise. For a set $X$, we often identify $X^S$ with $X^\Ss$, $X^A$ with $X^\As$, $X^{S\times A}$ with $X^{\Ss\times \As}$, and so on.

\subsection{MDP Preliminaries}

For a policy $\pi$, we define $P^{\pi} \in \mathbb{R}^{SA\times SA}$ to be the matrix with entries
\begin{align}
    P^{\pi}_{(s,a),(s',a')} = \begin{cases}
        p(s'|s,a) & \text{if } a' = \pi(s'),\\
        0 & \text{otherwise}.
    \end{cases}
\end{align}
We define $P \in \mathbb{R}^{SA\times S}$ to be the matrix with entries $P_{(s,a),s'} = p(s'|s,a)$ and, for fixed $(s,a)\in \Ss\times \As$, we define $p_{s,a}\in \mathbb{R}^S$ to be the vector with entries $p_{s,a}[s'] = p(s'|s,a)$.
The preceding definitions mean that, for any $u\in \mathbb{R}^S$, we have $(Pu)[s,a] = p_{s,a}\trans u$.

For $u\in \mathbb{R}^S$, we define $\sigma^2(u) \in \mathbb{R}^{SA}$ to be a vector with entries $\sigma^2(u)[s,a] \coloneqq \var[u[s'] \, | \, s'\sim p(\cdot\ | \ s,a)]$. Note that this means $\sigma^2(u) = Pu^2 - (Pu)^2$. Naturally, we write $\sigma(u) \coloneqq \sqrt{\sigma^2(u)}$. 

We define the \emph{value operator of policy} $\pi$, $\T{\pi}:\mathbb{R}^S\to\mathbb{R}^S$, by its mapping of $u\in \mathbb{R}^S$, defined entry-wise by
\begin{equation}
    \T{\pi}(u)[s] \coloneqq r(s,\pi[s])+\gamma \, p_{s,\pi[s]}\trans u.
\end{equation}
It can be readily verified that  $\T{\pi}$ (for any $\pi$) is monotonically increasing with respect to the element-wise order ($\leq$) on $\mathbb{R}^S$, is a $\gamma$-contraction with respect to the $l_\infty$-norm on $\mathbb{R}^S$, and has unique fixed point $v^{\pi}$.

For a vector $q\in \mathbb{R}^{SA}$, we also define $v(q)\in \mathbb{R}^{S}$ and $\pi(q) \in \As^{S}$ by $v(q)[s] = \max_{a}\{q[s,a]\}$ and $\pi(q)[s] = \argmax_a\{q[s,a]\}$ respectively. Note that this means $v(q)[s] = q[s,\pi(q)[s]]$.

Finally, for the total-variance technique, we will also need:
\begin{thm}\cite{AgarwalJiangKakadeSun_Reinforcement_2021, AzarMunosKappen_MdpGenerative_2012}
\label{thm:azar}
For any policy $\pi$, we have
\begin{equation}\label{eq:azar}
    \norm{(I - \gamma P^{\pi})^{-1}\sigma(v^{\pi})} \leq \sqrt{2}/\Gamma^{1.5}.
\end{equation}
\end{thm}

\subsection{Quantum Preliminaries}
\label{sec:quantum-prelim}

We now describe quantum oracles in more detail using standard quantum notation (Dirac notation). We briefly review this notation so that the following definitions make sense and refer readers to \cite{NielsenChuang_QuantumComputation_2000} for more information.

In Dirac notation, vectors $v \in \mathbb{C}^n$ are written as as $\ket{v}$, and called ``ket $v$''. The notation $\ket{i}$, with $i\in [n]$, is reserved for the $i$-th standard basis vector. $\ket{0}$ is also reserved for the $1$st standard basis vector when there is no conflict. A ket $\ket{i_1 i_2\dots i_M}$ with $i_j\in\{0,1\}$ is interpreted as the vector $\ket{i+1}\in\mathbb{C}^{2^M}$, where $i$ is the integer that is represented by $i_1\dots i_M$ in binary.

\begin{defn}[Quantum oracle encoding of functions and vectors]
Let $\Omega$ be a finite set of size $n$ and $u\in \mathbb{R}^{\Omega}$ (equivalently, $u: \Omega \rightarrow \mathbb{R}$) where, for all $i\in \Omega$, $u_i$ is represented by an $M$-bit string $\bar{u}_i$. A quantum oracle encoding $u$ is a unitary matrix $U_u: \mathbb{C}^n\otimes \mathbb{C}^{2^M} \rightarrow \mathbb{C}^n\otimes \mathbb{C}^{2^M}$ such that $U_u: \ket{i}\otimes \ket{0} \mapsto \ket{i}\otimes \ket{\bar{u}_i}$ for all $i\in [n]$.
\end{defn}

Like in the classical setting, we may always assume that $M$ is sufficiently large for our purposes.

\begin{defn}[Quantum oracle encoding of probability distributions]
Let $\Omega$ be a finite set of size $n$ and $p=(p_x)_{x\in \Omega}$ a discrete probability distribution on $\Omega$. The quantum oracle encoding of $p$ is a unitary matrix $U_p: \mathbb{C}^n\otimes \mathbb{C}^{J} \rightarrow \mathbb{C}^n\otimes \mathbb{C}^{J}$ such that $U_p: \ket{0}\otimes \ket{0}  = \sum_{x\in \Omega} \sqrt{p_x} \, \ket{x}\otimes \ket{v_{s'}}$, where $0\leq J \in \mathbb{Z}$ is arbitrary and $\ket{v_{s'}}\in \mathbb{C}^J$ is arbitrary. 
\end{defn}

\begin{defn}[Quantum generative model of an MDP]\label{def:generative_oracle} The quantum generative model of an MDP, with transition probabilities $p(s'|s,a)$, is a unitary matrix $\Oracle: \mathbb{C}^S\otimes \mathbb{C}^A \otimes \mathbb{C}^S \otimes \mathbb{C}^J\rightarrow \mathbb{C}^S\otimes \mathbb{C}^A \otimes \mathbb{C}^S \otimes \mathbb{C}^J$ such that
\begin{equation}\label{eq:generative_oracle}
\begin{aligned}
    &\Oracle: \ket{s} \otimes \ket{a} \otimes \ket{0} \otimes \ket{0}\\
    &\mapsto \ket{s} \otimes \ket{a} \otimes \Big(\sum_{s'\in \Ss}\sqrt{p(s'|s,a)}\, \ket{s'}\otimes \ket{\psi_{s',s,a}}\Big),
\end{aligned}
\end{equation}
where $0\leq J \in \mathbb{Z}$ is arbitrary and
$\ket{\psi_{s',s,a}}\in \mathbb{C}^J$ is arbitrary.
\end{defn}

We stress that the quantum state output by $\Oracle$ in \cref{eq:generative_oracle} is analogous to a \textit{sample} drawn from the classical probability distribution $\{p(s'|s,a)\}_{s'\in \Ss}$ as opposed to that distribution fully written out on a piece of paper. In \cref{sec:app_oracle_construction}, we describe how to systematically and efficiently construct the quantum generative model from a circuit for a classical generative model. This construction is already implicit in, for example, \cite{Montanaro_MonteCarlo_2015, Hamoudi_Chebyshev_2019, Belovs_Quantum_2019}, but we provide a description for completeness.

\section{Analysis of Quantum Algorithms}

In this section, we formally analyze our two algorithms \SolveMdpI\ and \SolveMdpII. These algorithms make essential use of two quantum subroutines: quantum mean estimation and quantum maximum finding. We begin by specifying the performance guarantees of these subroutines.

\subsection{Quantum Mean Estimation and Maximum Finding}

\begin{thm}[Quantum mean estimation~\cite{Brassard_AmplitudeEstimation_2000,Montanaro_MonteCarlo_2015}]~\label{thm:montanaro_monte_carlo_mean}
There are two quantum algorithms \AI\ and \AII\ with the following specifications. Let $\Omega$ be a finite set, $p = (p_x)_{x\in \Omega}$ a discrete probability distribution on $\Omega$, and function $v: \Omega \rightarrow \mathbb{R}$. Given quantum oracles $U_p$ and $U_v$  encoding $p$ and $v$ respectively. Then, 
\begin{enumerate}[topsep=0pt,itemsep=0ex,partopsep=1ex,parsep=1ex]
\item \AI\ requires $u,\epsilon>0$ as additional inputs and a promise $\zero\leq v \leq u$, in which case $\AI$ uses $O(u/\epsilon + \sqrt{u/\epsilon})$ queries to $U_p$, alternatively
\item \AII\ requires $\sigma>0$ and $\epsilon \in (0,4\sigma)$ as additional inputs and a promise $\var[v(x)\, | \, x\sim p] \leq \sigma^2$, in which case $\AII$ uses $O((\sigma/\epsilon) \log^{2}(\sigma/\epsilon))$ queries to $U_p$
\end{enumerate}
to output an estimate $\hat{\mu}'$ of $\mu \coloneqq \expect[v[x]\, | \, x \sim p] = p\trans v$ with
$\pr(\abs{\hat{\mu}'-\mu} > \epsilon) < 1/3$. Moreover, by repeating one of $\AI$ or $\AII$ $O(\log(1/\delta))$ times and taking the median output yields an estimate $\hat{\mu}$ of $\mu$ with $\pr(\abs{\hat{\mu}-\mu} < \epsilon) > 1-\delta$.
\end{thm}

For $i\in\{1,2\}$, we write $\qEsti_{\delta}(p\trans v,\epsilon)$ for an estimate of the mean of $v[x]$, with $x$ distributed as $p$, to error $<\epsilon$ with probability $>1-\delta$, using $\qEsti$. 

The median-of-means part of \Cref{thm:montanaro_monte_carlo_mean} is sometimes referred to as the ``powering lemma''~\cite{Jerrum_Powering_1986}.

\begin{thm}[Quantum maximum finding~\cite{Durr_MinFinding_1996}]~\label{thm:min_finding}
There exists a universal constant $c_{max}>0$ such that the following holds. There is a quantum algorithm $\qargmax$ such that, given a quantum oracle $U_u$ encoding a vector $u\in \mathbb{R}^n$, $\mathcal{A}_{\max}$ at most $c_{\max}\sqrt{n}\log(1/\delta)$ queries to $U_u$ and finds $\mathrm{argmax}_i(u_i)$ with probability $>1-\delta$.
\end{thm}

We write $\qargmax_{\delta}\{u[i]: i\in [n]\}$ for an estimate of the maximum of $u$, with probability $>1-\delta$, using $\qargmax$.

\subsection{Analysis of \normalfont{\SolveMdpI}}

We will use the following lemma which clearly follows from the if-then-else statement appearing in \SolveMdpI.
\begin{lem}\label{lem:if_then_else_guarantees}
For all $k\in[K]$ and $l\in \{0\}\cup [L]$, the $v_{k,l}$s are monotone increasing with respect to $(k-1)L+l$.  Moreover, for all $k\in[K]$ and $l\in [L]$, we have $v_{k,l} \geq v(q_{k,l-1})$.
\end{lem}

Using \lemref{if_then_else_guarantees} and the fact that our mean estimates are always shifted down to have one-sided error, we can prove the following proposition similarly to \cite[Section E of arXiv version]{Sidford_NearOptimal_2018}; the key point is to show that $v_{k,l}\leq\T{v_{k,l}}(v_{k,l})$. We present the full details for completeness.
\begin{prop}\label{prop:upper_bound}
For all $k\in [K]$ and $l\in [L]$, we have
\begin{alignat}{3}
    &v_{k,l} &&\leq v^{\pi_{k,l}} &&\leq v^{*} \label{eq:upper_v},\\
    &q_{k,l} &&\leq q^{\pi_{k,l}} &&\leq q^{*} \label{eq:upper_q},
\end{alignat}\label{prop:upper_vq}
with probability at least $1-\delta$.
\end{prop}

\begin{proof}
We first consider the failure probability. As all estimations are carried out with maximum failure probability $f\coloneqq \delta/4KLSA$ and there are $3KSA + KLSA < 4KLSA$ estimations (Lines 7, 8 and 12), the probability that there exists an incorrect estimate (up to the specified error) is at most $\delta$ by the union bound. 

We henceforth assume the $\qEst$ steps are all correct and proceed to prove \cref{eq:upper_v} and \cref{eq:upper_q}.

The second inequalities in \cref{eq:upper_v} and \cref{eq:upper_q} are clear from the definitions of $v^*$ and $q^*$. We therefore only show the first inequalities below and refer to them when referring to \cref{eq:upper_v} and \cref{eq:upper_q}. The main idea is to use \lemref{if_then_else_guarantees} together with the inequalities
\begin{align}
    x_k &\leq Pv_{k,0},\label{eq:base_monotone_anchor}
    \\
    \Delta_{k,l} &\leq Pv_{k,l} - Pv_{k,0},\label{eq:base_monotone_difference}
\end{align}
that are immediate from the definitions of $x_k$ and $\Delta_{k,l}$ on Lines 8 and 12 respectively because the subtracted terms equal the estimation errors.

To show \cref{eq:upper_v}, it suffices to show
\begin{equation}\label{eq:monotonicity}
    v_{k,l} \leq \T{\pi_{k,l}}(v_{k,l}).
\end{equation}
\Cref{eq:upper_v} then follows from repeatedly applying $\T{\pi_{k,l}}$ on both sides of \cref{eq:monotonicity}, and using the fact that $\T{\pi_{k,l}}$ is monotone increasing and is a contraction with unique fixed point $v^{\pi_{k,l}}$.

We proceed to show \cref{eq:monotonicity} by induction on $n\coloneqq(k-1)L+l$. The base case $n=0$ is true because $v_{1,0}\coloneqq \zero \leq \T{\pi_{1,0}}(v_{1,0}) = r$. The case $n=1$ is also true because $v_{1,1} = v(q_{1,0}) = \zero \leq \T{\pi_{1,1}}(v_{1,1}) = r$, where we used $q_{1,0}\coloneqq \zero$. In addition, note that $v_{k,L}\leq \T{\pi_{k,L}}(v_{k,L})$ is the same as $v_{k+1,0}\leq \T{\pi_{k+1,0}}(v_{k+1,0})$ by definitions on Line 15. This means that once we have established the truth of \cref{eq:monotonicity} at $k=k', l=L$, we can assume its truth at $k=k'+1, l=0$.

Now consider $n>1$. We prove \cref{eq:monotonicity} element-wise for each $s\in \Ss$ by considering the following two cases that could happen at the if-clause on Line 10. 
\begin{enumerate}
\item Case $v(q_{k,l-1})[s] \geq v_{k,l-1}[s]$. Then
\begin{equation}
    \begin{aligned}
        v_{k,l}[s] &\coloneqq v(q_{k,l-1})[s]
        \\
        &= q_{k,l-1}[s,\pi_{k,l}[s]]
        \\
        &= \max\{r[s,\pi_{k,l}(s)] + \gamma (x_{k}[s,\pi_{k,l}[s]]+ \Delta_{k,l-1}[s,\pi_{k,l}[s]]), 0\}
        \\
        &\leq r[s,\pi_{k,l}(s)] + \gamma (Pv_{k,l-1})[s,\pi_{k,l}(s)]
        \\
        &=\T{\pi_{k,l}}(v_{k,l-1})[s]
        \\
        &\leq \T{\pi_{k,l}}(v_{k,l})[s],
\end{aligned}
\end{equation}
where the second line uses $\pi_{k,l}[s] \coloneqq \pi(q_{k,l-1})[s]$ in this case, the third line uses definition of $q_{k,l-1}$ (for $n>1$), the fourth line uses \cref{eq:base_monotone_anchor} and \cref{eq:base_monotone_difference} and $0 \leq v_{k,l-1}$ (\lemref{if_then_else_guarantees}) to remove the $\max$, and the last line uses $v_{k,l-1}\leq v_{k,l}$ (\lemref{if_then_else_guarantees}).

\item Case $v(q_{k,l-1})[s] < v_{k,l-1}[s]$. Then 
\begin{equation}
    v_{k,l}[s] \coloneqq v_{k,l-1}[s]\leq \T{\pi_{k,l-1}}(v_{k,l-1})[s] \leq \T{\pi_{k,l-1}}(v_{k,l})[s] = \T{\pi_{k,l}}(v_{k,l})[s],
\end{equation}
where the first inequality is by the inductive hypothesis, the second inequality uses $v_{k,l-1}\leq v_{k,l}$ (\lemref{if_then_else_guarantees}), and the last equality uses $\pi_{k,l}[s] \coloneqq \pi_{k,l-1}[s]$ in this case.
\end{enumerate}
Therefore, we have established \cref{eq:monotonicity}, and so \cref{eq:upper_v}. 

\Cref{eq:upper_q} then follows from 
\begin{equation}
    q_{k,l} \leq r + \gamma Pv_{k,l} \leq r + \gamma Pv^{\pi_{k,l}} =q^{\pi_{k,l}}, 
\end{equation}
where the first inequality again uses \cref{eq:base_monotone_anchor} and \cref{eq:base_monotone_difference} and $\zero \leq v_{k,l}$ (\lemref{if_then_else_guarantees}), and the second inequality uses \cref{eq:upper_v} which we have just established.
\end{proof}

The above proposition shows that $v_{k,L}$ and $q_{k,L}$ are upper bounded by $v^*$ and $q^*$ respectively. Therefore, the following proposition shows that $v_{k,L}$ and $q_{k,L}$ are converging to $v^*$ and $q^*$ respectively.
\begin{prop}\label{prop:lower_bound}
For all $k\in [K]$, we have
\begin{align}
v^*-\epsilon_k &\leq v_{k,L},\label{eq:lower_v}\\
q^* - \epsilon_k &\leq q_{k,L}\label{eq:lower_q},
\end{align}
with probability at least $1-\delta$.
\end{prop}
If there were no mean estimation errors, \propref{lower_bound} follows from the contractive properties of the Bellman operator. The challenge for us is to analyze those errors carefully. As we mentioned in our Introduction, the errors involved here go beyond those analyzed in \cite{Sidford_NearOptimal_2018}.  
\begin{proof}
By reusing the first paragraph in the proof of \propref{upper_bound}, we can readily set aside consideration of the failure probability. We henceforth assume  the $\qEst$ steps are all correct and proceed to prove \cref{eq:lower_v} and \cref{eq:lower_q}.

We proceed by induction on $k\geq 0$  with the inductive hypothesis comprising both inequalities above for all indices strictly less than $k$. The base case $k=0$ can be established by defining $\epsilon_0\coloneqq \Gamma$, $v_{0,L}\coloneqq \bf{0}$, and $q_{0,L}\coloneqq \bf{0}$. Note that these definitions are consistent with the induction steps below.

Now consider $k>0$. The main idea is to use \cref{thm:azar} and the inequalities
\begin{align}
    x_k &\geq Pv_{k,0} -2c(1-\gamma)^{1.5}\epsilon \sqrt{y_k+b},\\
    \Delta_{k,l} &\geq Pv_{k,l} - Pv_{k,0} -2c(1-\gamma)\epsilon_k,
\end{align}
that are immediate from the definitions of $x_k$ and $\Delta_{k,l}$ on Lines 8 and 12 respectively.

We first show \cref{eq:lower_q}. Define vector $\xi_{k}\in\mathbb{R}^{SA}$ by
\begin{equation}\label{eq:xi_definition}
    \xi_{k}\coloneqq 2c(1-\gamma)^{1.5}\epsilon\sqrt{y_k+b}+2c(1-\gamma)\epsilon_k,
\end{equation}
then we have
\begin{equation}\label{eq:upper_bound_recursion}
\begin{aligned}
    q^* - q_{k,l} 
    &= r + \gamma P^{\pi^*}q^* - \max\{r + \gamma (x_k + \Delta_{k,l}),\zero\}
    \\
    &\leq \gamma P^{\pi^*}q^* - \gamma (x_k+ \Delta_{k,l})
    \\
    &\leq \gamma P^{\pi^*}q^* - \gamma (\cancel{Pv_{k,0}} + Pv_{k,l} - \cancel{Pv_{k,0}} - 2c(1-\gamma)^{1.5}\epsilon\sqrt{y_k+b}-2c(1-\gamma)\epsilon_k)
    \\
    &\leq \gamma P^{\pi^*}q^* - \gamma Pv_{k, l} + 2c(1-\gamma)^{1.5}\epsilon\sqrt{y_k+b}+2c(1-\gamma)\epsilon_k
    \\
    &= \gamma P^{\pi^*}q^* - \gamma Pv_{k, l} + \xi_{k}
    \\
    &\leq \gamma P^{\pi^*}q^* - \gamma Pv(q_{k, l-1}) + \xi_{k}
    \\
    &\leq \gamma P^{\pi^*}(q^* - q_{k, l-1}) + \xi_k,
\end{aligned}
\end{equation}
where the fourth line uses $\gamma\leq 1$, the sixth line uses $v(q_{k,l-1})\leq v_{k,l}$ (\lemref{if_then_else_guarantees}), and the last line uses $P^{\pi^*}q_{k,l-1}\leq Pv(q_{k,l-1})$ which follows from definitions.

Recursing \cref{eq:upper_bound_recursion} with respect to $l\geq 1$ gives
\begin{equation}\label{eq:overall_bound}
\begin{aligned}
    q^* - q_{k,l} &\leq \gamma^{l} (P^{\pi^*})^{l}(q^* - q_{k, 0}) + \sum_{i=0}^{l-1}\gamma^{i}(P^{\pi^*})^i \xi_k 
    \\[2pt]
    &\leq \gamma^l \Gamma + (I-\gamma P^{\pi^*})^{-1}\xi_k,
\end{aligned}
\end{equation}
where the last line uses $q^*-q_{k,0} \leq q^* \leq \Gamma$ as $q_{k,0}\geq \zero$ by definitions on Line 4 and Line 13. The first term, $\gamma^{l}\Gamma$, can be bounded when $l=L-1,L$:
\begin{equation}\label{eq:gamma_h_bound}
    \gamma^{L}\Gamma \leq \gamma^{L-1}\Gamma \leq \exp(-(L-1)(1-\gamma)) \Gamma \leq \epsilon/4 \leq \epsilon_k/2,
\end{equation}
where the second inequality uses $x\leq \exp(-(1-x))$ for all $x\in \mathbb{R}$, the third inequality uses the definition $L\coloneqq \Gamma \ceil{\log (4\Gamma/\epsilon)}+1$, and the last inequality uses $\epsilon \leq 2\epsilon_K\leq 2\epsilon_k$ for all $k\in[K]$ which follows from $K \leq \log_2(\Gamma/\epsilon)+1$.

We now bound the second term, $(I-\gamma P^{\pi^*})^{-1}\xi_k$. To this end, we first bound the term $\sqrt{y_k+b}$ appearing in $\xi_k$. From the definition of $y_k$, there exists a $b'$ with $\abs{b'}\leq b$ such that
\begin{equation}
\begin{aligned}
    \sqrt{y_k+b} &\leq \max \{(P v_{k,0}^2 + b - (Pv_{k,0} + (1-\gamma)b')^2)^{1/2},\sqrt{b}\}
    \\
    &\leq (\sigma^2(v_{k,0}) + b + 2(1-\gamma)|b'|Pv_{k,0})^{1/2}
    \\
    &\leq \sqrt{\sigma^2(v_{k,0})+ 3b}
    \\
    &\leq \sigma(v_{k,0}) + \sqrt{3b}
    \\
    &\leq \sigma(v^*) + \sigma(v^*-v_{k,0}) + \sqrt{3b},
\end{aligned}
\end{equation}
where the second line uses $\zero\leq v_{k,0}$ (\lemref{if_then_else_guarantees}) to remove the $\max$, the third line uses $v_{k,0}\leq \Gamma$ (\propref{upper_bound}), and
the last line uses the fact that, for any random variables $X$ and $Y$, we have $\var[X+Y] = \var[X] + \var[Y] + 2 \, \mathrm{Cov}[X, Y] \leq (\sqrt{\var[X]} + \sqrt{\var[Y]})^2$.

But we have $v_{k,0}-v^*\leq 0$ from \cref{eq:upper_v} of \propref{upper_vq} and $v^*-v_{k,0} = v^* - v_{k-1,L} \leq \epsilon_{k-1}$ by the inductive hypothesis. Therefore, $\sigma(v_{k,0}-v^*) \leq \norm{v_{k,0}-v^*} \leq \epsilon_{k-1}= 2\epsilon_k$, and therefore 
\begin{equation}
    \sqrt{y_k+b}\leq  \sigma(v^*) + 2\epsilon_k + \sqrt{3b}.
\end{equation}
Therefore, recalling $\xi_k \coloneqq 2c(1-\gamma)^{1.5}\epsilon\sqrt{y_k+b} + 2c(1-\gamma)\epsilon_k$ from \cref{eq:xi_definition}, we have
\begin{equation}\label{eq:h_xi_bound}
\begin{aligned}
    (I - \gamma P^{\pi^*})^{-1}\xi_{k} &= 2c(1-\gamma)^{1.5}\epsilon (I - \gamma P^{\pi^*})^{-1} \sqrt{y_k + b} + 2c(1-\gamma) \epsilon_k(I - \gamma P^{\pi^*})^{-1}\one
    \\
    & \leq 2c (1-\gamma)^{1.5}\epsilon (I - \gamma P^{\pi^*})^{-1} (\sigma(v^*)+2\epsilon_k+\sqrt{3b}) + 2c (1-\gamma) \epsilon_k(I - \gamma P^{\pi^*})^{-1}\one
    \\
    &\leq 2c(1-\gamma)^{1.5}\epsilon (I - \gamma P^{\pi^*})^{-1} \sigma(v^*) + 2c\sqrt{1-\gamma}\,\epsilon\, 2\epsilon_k + 2c\sqrt{1-\gamma}\,\epsilon\, \sqrt{3b} + 2c\,\epsilon_k
    \\
    &\leq 2c \, \sqrt{2}\,\epsilon + 2c \,\sqrt{1-\gamma}\, \epsilon\, 2\epsilon_k + 2c \,\epsilon \sqrt{3b} + 2c\, \epsilon_k
    \\
    &\leq 2c \,(2\sqrt{2}+ 2 + 2\sqrt{3b} + 1)\epsilon_k
    \\
    &< \epsilon_k/2,
\end{aligned}
\end{equation}
where the third line uses $(I-\gamma P^{\pi^*})^{-1} \one \leq (1-\gamma)^{-1}$, the fourth line \emph{crucially} uses \cref{thm:azar} with $\pi$ set to $\pi^*$, the fifth line uses $\epsilon \leq 2\epsilon_k$ for all $k\in[K]$ and the input assumption $\sqrt{1-\gamma}\,\epsilon\leq1$, i.e., $\epsilon\leq \sqrt{\Gamma}$, and the last line uses definitions $b \coloneqq 1$ and $c\coloneqq 0.01$.

Using \cref{eq:gamma_h_bound} and \cref{eq:h_xi_bound} to bound the first and second terms in \cref{eq:overall_bound} respectively, we find
\begin{alignat}{2}
    &q^* - q_{k,L} &&\leq \epsilon_k,
    \\
    &q^* - q_{k,L-1} &&\leq \epsilon_k.
\end{alignat}
The top equation is one inequality we wish to show in our induction. The bottom equation can be used to establish the other inequality as follows. For all $s\in \Ss$, we have
\begin{equation}
v_{k,L}[s] \geq v(q_{k,L-1})[s] = \max_{a}\{q_{k,L-1}[s,a]\} \geq \max_{a}\{q^*[s,a]-\epsilon_k\} = v^*[s]- \epsilon_k,
\end{equation}
where the first inequality is by \lemref{if_then_else_guarantees}. Hence $v_{k,L}\geq v^* - \epsilon_k$, as desired.
\end{proof}

The correctness of \algoref{SolveMdp1} then follows from combining \propref{upper_bound} and \propref{lower_bound} with $k=K$ and $l=L$ and recalling the definitions of $(\hat{v},\, \hat{\pi},\, \hat{q})$ and $K$. Formally:
\begin{thm}[Correctness of \normalfont{\SolveMdpI}]\label{thm:correctness_1}
The outputs $\hat{v}$, $\hat{\pi}$, and $\hat{q}$ of \SolveMdpI\ satisfy
\begin{alignat}{4}
    &v^* - \epsilon && \leq \hat{v} && \leq v^{\hat{\pi}} &&\leq v^*,
    \\
    &q^* - \epsilon && \leq \hat{q} && \leq q^{\hat{\pi}} && \leq q^*,
\end{alignat}
with probability at least $1-\delta$.
\end{thm}

Having shown correctness, we turn to complexity:
\begin{thm}[Complexity of \normalfont{\SolveMdpI}]
\label{thm:complexity_1}
The quantum query complexity of \SolveMdpI\ is
\begin{equation}
\Order{SA(\Gamma^{1.5}\, \epsilon^{-1} + \Gamma^2)\, \log^4(\Gamma/\epsilon) \log(SA\Gamma/\delta)}.
\end{equation}
\end{thm}
The proof of \Cref{thm:complexity_1} involves showing \Cref{thm:montanaro_monte_carlo_mean} is applicable and applying it.

\begin{proof}
    As in the correctness analysis, we assume that all estimations are correct, up to the specified error, because the probability that this does not hold is at most $\delta$. This means we can assume all results obtained during the correctness analysis. In the following, we will use $K=O(\log(\Gamma/\eps))$ and $L = O(\Gamma \log (\Gamma/\epsilon))$ without further remarks. 

Let $C$ be the complexity of \SolveMdpI\ as if all estimations were carried out with maximum failure probabilities set to constant. Then, since the actual maximum failure probabilities are set to $f \coloneqq \delta/4KLSA$, the actual complexity of \SolveMdpI\ is
\begin{equation}\label{eq:amplification_factor}
    O(C \log(KLSA/\delta)) = O(C \log(SA\Gamma\log(\Gamma/\epsilon)/\delta)).
\end{equation}

Now we bound $C$ by examining each line involving $\qEst$ in turn and using \cref{thm:montanaro_monte_carlo_mean}. 

On Line 7, we can bound $\zero\leq v_{k,0}\leq v^* \leq \Gamma$. Therefore, we can use quantum mean estimation algorithm $\AI$ in \cref{thm:montanaro_monte_carlo_mean}, which results in an overall query cost of order
\begin{equation}\label{eq:var_est_cost}
    SA K (\Gamma^2 \, b^{-1} + \sqrt{\Gamma^2 \, b^{-1}} + \Gamma(1-\gamma)^{-1}b^{-1} + \sqrt{\Gamma(1-\gamma)^{-1}b^{-1}}) = O(SA \Gamma^2\log(\Gamma/\epsilon)).
\end{equation}

On Line 8, we see that $\sigma^2(v_{k,0})[s,a] \leq \sqrt{y_k[s,a] + b}$. We also note that $0< (1-\gamma)^{1.5}\epsilon \sqrt{y_k[s,a] + b} < 4 \sqrt{y_k[s,a] + b}$. Therefore, we can use quantum mean estimation algorithm $\AII$ in \cref{thm:montanaro_monte_carlo_mean}, with error set to $(1-\gamma)^{1.5}\epsilon \sqrt{y_k[s,a] + b}$ and variance upper bound set to $y_k[s,a] + b$, which results in an overall query cost of order 
\begin{equation}\label{eq:anchor_cost}
    K\sum_{(s,a)\in\Ss\times\As} \, 
    w[s,a] \log^2(w[s,a]) = O(SA\Gamma^{1.5}\epsilon^{-1}\log^3(\Gamma/\epsilon)),
\end{equation}
where, importantly, $w[s,a] \coloneqq \big(\cancel{\sqrt{y_k[s,a] + b}}\big)\big({(1-\gamma)^{1.5}\epsilon \, \cancel{\sqrt{y_k[s,a] + b}}}\big)^{-1} = \Gamma^{1.5}/\epsilon$.

On Line 12, we can bound $\zero \leq v_{k,l}-v_{k,0}\leq v^* - v_{k,0} \leq \epsilon_{k-1} = 2\epsilon_k$. Therefore, we can use quantum mean estimation algorithm $\AI$ in \cref{thm:montanaro_monte_carlo_mean}, which results in an overall cost of order
\begin{equation}\label{eq:shift_cost}
    LSA \, \Bigg(\frac{2 \cancel{\epsilon_k}}{c(1-\gamma)\cancel{\epsilon_k}} + \sqrt{\frac{2 \cancel{\epsilon_k}}{c(1-\gamma)\cancel{\epsilon_k}}}\Bigg) = O(SA \Gamma^2 \log(\Gamma/\epsilon)).
\end{equation}

Adding together \cref{eq:var_est_cost}, \cref{eq:anchor_cost}, and \cref{eq:shift_cost}, and noting that all logarithmic terms are at most $\log^3(\Gamma/\epsilon)$, shows that 
\begin{equation}
    C = O(SA(\Gamma^{1.5}\epsilon^{-1} + \Gamma^2)\log^3(\Gamma/\epsilon)).
\end{equation}
Combining the above equation with \cref{eq:amplification_factor} shows that the overall quantum query complexity of \SolveMdpI\ is
\begin{equation}
O(SA(\Gamma^{1.5}\epsilon^{-1} + \Gamma^2)\log^3(\Gamma/\epsilon)\log(SA\Gamma\log(\Gamma/\eps)/\delta)) = O(SA(\Gamma^{1.5}\epsilon^{-1} + \Gamma^2)\log^4(\Gamma/\epsilon)\log(SA\Gamma/\delta)),
\end{equation}
as desired.
\end{proof}

\subsection{Analysis of \normalfont{\SolveMdpII}}

If we only require $v^*$ and $\pi^*$ but not $q^*$, then we present an alternative quantum algorithm we call \SolveMdpII\ that is quadratically faster than \SolveMdpI\ in terms of $A$. The source of the speedup in $A$ is our use of quantum maximum finding to find the maximum at each iteration $l$, $\qargmax_f\{q_{l-1,s}[a]\ : \ a\in \As\}$ on Line 6 which uses the quantum oracle encoding $U_{q_{l-1},s}$ created in the previous iteration $l-1$.

Failure probability aside, our strategy for proving the correctness of \SolveMdpII\ (\cref{thm:correctness_2}) is to observe the similarity between \SolveMdpII\ and \SolveMdpI\ and then reuse the arguments used to prove the correctness of \SolveMdpI.

\SolveMdpII\ is similar to \SolveMdpI\ with $k$ set to $1$. In particular, the vectors $z_{l}, \, q_{l}\in \mathbb{R}^{SA}$, defined entry-wise by
\begin{align}
    z_{l}[s,a] &\coloneqq z_{l,s}[a],\\
    q_{l}[s,a] &\coloneqq q_{l,s}[a],
\end{align}
are analogous to the vectors $x_{1} + \Delta_{1,l}$ and $q_{1,l}$ appearing in \SolveMdpI\ respectively. Moreover, the $\tilde{v}_l[s] \coloneqq \max_{a}\{q_{l-1,s}[a]\} = \max_{a}\{q_{l-1}[s,a]\}$ appearing in \SolveMdpII\ corresponds exactly to the $v(q_{1,l-1})[s] \coloneqq \max_{a}\{q_{1,l-1}[s,a]\}$ appearing in \SolveMdpI.  

Having observed the similarity between \SolveMdpII\ and \SolveMdpI, the following analogue of \lemref{if_then_else_guarantees} due to the if-then-else statement is clear.
\begin{lem}\label{lem:if_then_else_guarantees_2}
For all $l\in [L]$, the $v_l$s are monotone increasing, that is $v_{l-1} \leq v_{l}$, and moreover we have $v_{l} \geq v(q_{l-1})$.
\end{lem}

We now proceed to establish the correctness and complexity of \SolveMdpII. In the proof of correctness, we will reuse the proofs of \propref{upper_bound} and \propref{lower_bound} unchanged except that they now invoke \cref{lem:if_then_else_guarantees_2} instead of \cref{lem:if_then_else_guarantees}.
\begin{thm}[Correctness of \SolveMdpII.]\label{thm:correctness_2}
The outputs $\hat{v}$ and $\hat{\pi}$ of \SolveMdpII\ satisfy
\begin{equation}\label{eq:sandwich_v_2}
    v^*-\epsilon \leq \hat{v} \leq v^{\hat{\pi}} \leq v^*,
\end{equation}
with probability at least $1-\delta$. 
\end{thm}

\begin{proof}
    We first consider the failure probability. The analysis is similar to that used to prove \propref{upper_bound} except that we now need to analyze quantum oracles that may fail. To do this, we appeal to basic facts about unitary matrices, in particular, a quantum version of the union bound stating that the failure probabilities of quantum operators, i.e., unitary matrices, add linearly. On Line 10, because $U_{z_{l,s}}$ is created using $\qEst$ with failure probability $f$, it is $2Af$-close to its ``ideal version''. More precisely, we mean that there exists a quantum oracle $U^{\text{ideal}}_{z_{l,s}}$ encoding $\widehat{(Pv_l)}[s,a]-(1-\gamma)\epsilon/4$, where $\widehat{(Pv_l)}[s,a]$ satisfies $|\widehat{(Pv_l)}[s,a]-(Pv_l)[s,a]|\leq (1-\gamma)\epsilon/4$, such that $\norm{U^{\text{ideal}}_{z_{l,s}}-U_{z_{l,s}}}_{\mathrm{op}}\leq 2Af$. Since $U_{q_{l,s}}$ can be created using one call to $U_{z_{l,s}}$ and one call to $U_{z_{l,s}}^{-1}$, it is $4Af$-close to its ideal version (defined similarly). Then, on Line 6, $\qargmax$ uses the oracle $U_{q_{l,s}}$ at most $c_{\max}\sqrt{A}\log(1/\delta)$ times. By the quantum union bound and substituting in the definition of $f$, this means the quantum operation implemented by $\qargmax$ is $(c_{\max}\sqrt{A}\log(1/\delta)\cdot 4Af = \delta/LS)$-close to its ideal version. This means that the output of $\qargmax$ is incorrect with probability at most $\delta/LS$. Since $\qargmax$ is invoked a total of $LS$ times, we see that the overall probability of failure is at most $\delta$ by the (usual) union bound. 

We henceforth assume the \qEst\ and $\qargmax$ steps are all correct and proceed to prove \cref{eq:sandwich_v_2}.

The last inequality, $v^{\hat{\pi}}\leq v^*$, is clear.

To prove the middle inequality, $\hat{v}\leq v^{\hat{\pi}}$, we can directly reuse the proof of \propref{upper_bound} provided we have $z_{l} \leq Pv_l.$ But this is clear because $x_{l}$ is equal to an estimate of $Pv_l$ with the estimation error subtracted off.

To prove the first inequality, $v^*-\epsilon \leq \hat{v}$, we can reuse the proof of \propref{lower_bound}, provided we have $z_l \geq Pv_l - (1-\gamma)\epsilon/2$, which is true. Defining $\xi \coloneqq (1-\gamma)\epsilon/2 \cdot\one\in \mathbb{R}^{SA}$, we see from the proof of \propref{lower_bound} that
\begin{equation}
    q^* - q_{L-1 }\leq \gamma^{L-1} \Gamma + (1-\gamma P^{\pi^*})^{-1}\xi \leq \epsilon,
\end{equation}
since $L \coloneqq \Gamma \ceil{\log(4\Gamma/\epsilon)}+1$. Therefore, for all $s\in \Ss$, we have
\begin{equation}
v_L[s] \geq v(q_{L-1})[s] = \max_{a}\{q_{L-1}[s,a] \geq \max_a\{q^*[s,a]-\epsilon\} = v^*[s] - \epsilon.
\end{equation}
\end{proof}

\begin{thm}[Complexity of \normalfont{\SolveMdpII}]
\label{thm:complexity_2}
The quantum query complexity of \SolveMdpII\ is
\begin{equation}
    \Order{S\sqrt{A} \, \Gamma^3\, \epsilon^{-1}  \, \log^2(\Gamma/\epsilon) \log(SA\Gamma/\delta)}.
\end{equation}
\end{thm}

\begin{proof}
We can assume all results obtained during the correctness proof of \SolveMdpII.

We let $C$ be the complexity of \SolveMdpII\ as if all estimations and maximum finding were carried out with maximum failure probabilities set to constant. Then the actual complexity of our algorithm is
\begin{equation}\label{eq:amplification_factor_2}
    O(C\log(LSA/\delta)) = O(C\log(SA\Gamma\log(\Gamma/\eps)/\delta)),
\end{equation}
since the actual maximum failure probabilities are set to $f \coloneqq \delta/4c_{\max}LSA^{1.5}\log(1/\delta)$ and $L = O(\Gamma \log (\Gamma/\epsilon))$.

Now we bound $C$. Note that, for all $l\in [L]$, we have
\begin{equation}
    \zero\leq v_l \leq v^* \leq \Gamma.
\end{equation}
By using $\AI$ of \cref{thm:montanaro_monte_carlo_mean} to do the $\qEst$ on Line 10, the query complexity of $U_{z_{l,s}}$ is 
\begin{equation}
    \frac{\Gamma}{(1-\gamma)\epsilon/4} + \sqrt{\frac{\Gamma}{(1-\gamma)\epsilon/4}} = O(\Gamma^2/\epsilon),
\end{equation}
provided $\epsilon = O(\Gamma^2)$. But we have (trivially) assumed $\epsilon\leq \Gamma$ on the input $\epsilon$, so this holds.

As $U_{q_{l,s}}$ uses one call to $U_{z_{l,s}}$ and one call to its inverse $U_{z_{l,s}}^{-1}$, the query complexity of $U_{q_{l,s}}$ is twice that of $U_{z_{l,s}}$. 

By means of the quantum maximum finding algorithm (\cref{thm:min_finding}) we only incur a multiplicative factor of  $O(\sqrt{A})$ when we invoke $\qargmax$ over an action space of size $A$. That is, for each $l\in [L]$ and $s\in \Ss$,  $\qargmax$ makes $O(\sqrt{A})$ queries to $U_{q_{l,s}}$ to find $\argmax_{a}\{q_{l-1,s}[a]\}$. There are also $L$ iterations, so
\begin{equation}
    C = O(L \, S\sqrt{A}\, \Gamma^2 \eps^{-1}) = O(S\sqrt{A}\, \Gamma^3  \eps^{-1} \, \log(\Gamma/\epsilon)),
\end{equation}
because $L = O(\Gamma \log(\Gamma/\epsilon))$. Combining the above equation with \cref{eq:amplification_factor_2} shows that the overall quantum query complexity of \SolveMdpII\ is 
\begin{equation}
O(S\sqrt{A}\, \Gamma^3  \eps^{-1} \, \log(\Gamma/\epsilon) \log(SA\Gamma\log(\Gamma/\eps)/\delta)) = O(S\sqrt{A}\, \Gamma^3  \eps^{-1} \, \log^2(\Gamma/\epsilon) \log(SA\Gamma/\delta)),
\end{equation}
as desired.
\end{proof}

\section{Lower Bounds}\label{sec:lower_bounds}

We now state our lower bounds on the number of samples needed to compute $q^*$, $v^*$, and $\pi^*$. Since our proof technique is very modular, we can prove lower bounds for both classical and quantum algorithms with only minor changes. 

Our classical lower bounds match known results \cite{AzarMunosKappen_MdpGenerative_2012,Sidford_NearOptimal_2018} and use a similar hard MDP instance, but they are qualitatively stronger as explained in the Introduction (end of \cref{sec:technical_overview}).

These lower bounds are interesting when the parameters $S$, $A$, and $\Gamma$ are large since the algorithms scale polynomially in these parameters. To avoid edge cases that make the analysis tedious, we only prove the lower bound for $S,A\geq 2$, and $\Gamma \geq 10$ (equivalently $\gamma \in [0.9,1)$).

\begin{thm}[Classical and quantum lower bounds]\label{thm:lower_bound}
    Fix any integers $S,A \geq 2$ and $\gamma \in [0.9,1)$. Let $\Gamma \coloneqq (1-\gamma)^{-1}\geq 10$ and fix any $\epsilon \in (0,\Gamma/4)$. There exists an MDP with $S$ states, $A$ actions, and discount parameter $\gamma$ such that the following lower bounds hold:
    \begin{enumerate}[topsep=0pt,itemsep=0ex,partopsep=1ex,parsep=1ex]
        \item Given access to a classical generative oracle, any algorithm that computes an $\eps$-approximation to $q^*$, $v^*$, or $\pi^*$ must make $\Omega(SA\Gamma^3/\eps^2)$ queries.
        \item Given access to a quantum generative oracle, any algorithm that computes an $\eps$-approximation to $q^*$ must make $\Omega(SA\Gamma^{1.5}/\eps)$ queries and any algorithm that computes an $\eps$-approximation to $v^*$ or $\pi^*$ must make $\Omega(S\sqrt{A}\Gamma^{1.5}/\eps)$ queries.
    \end{enumerate}
\end{thm}

We first establish the lower bound for an MDP with $S=2$ and $A=1$. Note that when $A=1$, there is only one action per state, so it is trivial to compute the optimal policy. So we can only show hardness for computing $q^*$ or $v^*$, which will be the same because there is only one action.

\begin{figure}[h]
    \centering
    \includegraphics[scale=0.55]{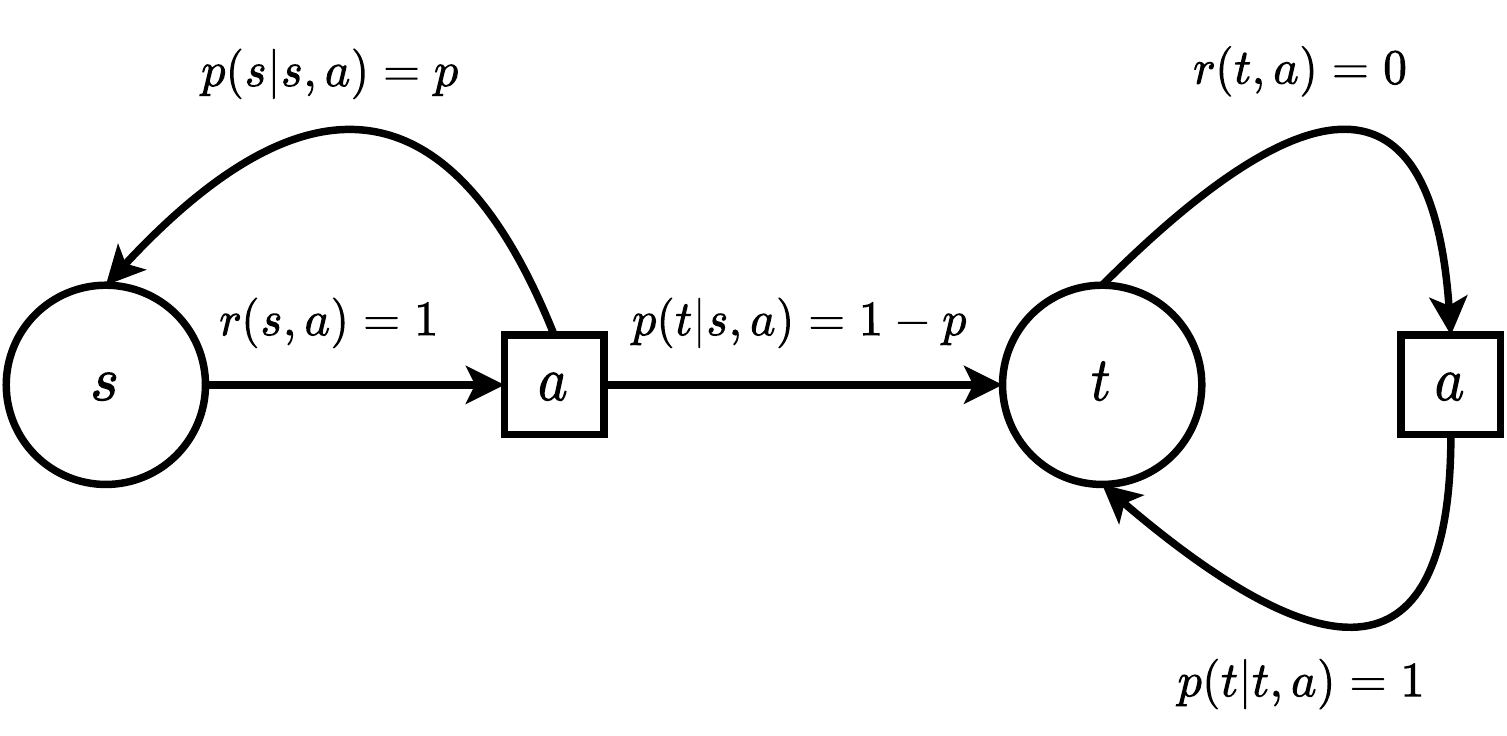}
    \caption{\label{fig:lowerbound}The MDP we use for the lower bound with $S=2$ and $A=1$. Distinguishing between $p\leq p_0$ and $p\geq p_0+\alpha$ is hard.}
\end{figure}

\begin{lem}\label{lem:two_state_lower_bound}
   Fix any $\gamma \in [0.9,1)$. Let $\Gamma \coloneqq (1-\gamma)^{-1} \geq 10$ and fix any $\epsilon \in (0,\Gamma/4)$. There exists an MDP shown in \Cref{fig:lowerbound} with 2 states and 1 action, for which computing $v^*$ (or equivalently, $q^*$) to error $\eps$ requires $\Omega(\Gamma^3/\eps^2)$ queries to a classical generative oracle or $\Omega(\Gamma^{1.5}/\eps)$ queries to a quantum generative oracle.
\end{lem}

\begin{proof}
    The MDP shown in \Cref{fig:lowerbound} has two states we call $s$ and $t$. State $t$ is a sink and the only transition from there is back to $t$ with no reward. Hence $v^*(t)=0$. State $s$ is a source, and on taking action $a$, there is a reward $r(s,a)=1$. The transition is probabilistic and controlled by an unknown probability $p\in(0,1)$. With probability $p$ we come back to $s$, and with probability $1-p$ we move to $t$. We can compute $v^*(s)$ using the equation $v^*(s) = 1 + \gamma (p v^*(s) + (1-p)v^*(t))$, which yields 
    \begin{equation}
        v^*(s)= \frac{1}{1-\gamma p}.
    \end{equation}
    Now further assume that we are promised that $p \leq p_0$ or $p \geq p_0 + \alpha$, where 
    \begin{equation}
        p_0 = 1 - \frac{1}{\Gamma} \quad \mathrm{and} \quad \alpha = \frac{3\eps}{\Gamma^2}.
    \end{equation}
    Note that $p_0+\alpha <1$ because of the way we have chosen the range of $\eps$.
    
    We claim that computing $v^*(s)$ to additive error $\eps$ will allow us to distinguish these two cases. To see this, note that the difference between the two values of $v^*(s)$ is at least
    \begin{equation}
    \begin{aligned}
        & \frac{1}{1-\gamma (p_0+\alpha)} - \frac{1}{1-\gamma p_0}\\
        = & \frac{\gamma \alpha}{(1-\gamma(p_0+\alpha))(1-\gamma p_0)}\\
        > & \frac{\gamma \alpha}{(1-\gamma p_0)^2} 
        \geq \frac{0.9 \alpha}{(1.1/\Gamma-1/10\Gamma^2)^2}\\
        \geq & {0.9 \alpha \Gamma^2/1.21}
        \geq \alpha \Gamma^2/1.35 \geq 2\eps.
    \end{aligned}
    \end{equation}
    Thus computing $v^*$ to additive error $\eps$ will allow us to distinguish these two possibilities.
    
    Now we just have to show that distinguishing a coin with probability of heads at most $p_0$ or at least $p_0 + \alpha$ given samples from this coin is as hard as claimed in the lower bound. We prove this via query complexity.
    
    Suppose that instead of having sample access to a coin, we have query access to an $n$-bit string $x$ with the promise that either at most $p_0$ fraction of its bits is equal to $1$ or at least $p_0 + \alpha$ fraction of its bits is equal to 1. Both quantumly and classically, we can query any bit $x_i$ of $x$ using $1$ query. It is easy to see that we can generate a sample from our coin with probability of heads equal to $|x|/n$ (the fraction of $1$s in $x$) with only $1$ query to $x$. This works both classically and quantumly.
    
    So we have shown a reduction from the problem of computing $v^*$ to error $\eps$ to the problem of deciding whether $|x|/n \leq p_0$  or $|x|/n \geq p_0 + \alpha$  given query access to an $n$-bit string $x$. This is the approximate counting problem. If we count the number of $0$s, we want to distinguish $1/\Gamma$ $0$s from $(1/\Gamma - 3\eps/\Gamma^2)$ $0$s. We need to approximate the count to multiplicative precision $O(\eps/\Gamma)$. Finally, we can invoke the known lower bounds for approximate counting summarized in \lemref{approxcount}. These give a classical lower bound of $\Omega(\Gamma^3/\eps^2)$ and a quantum lower bound of $\Omega(\Gamma^{1.5}/\eps)$ as claimed.    
\end{proof}

We formally state the approximate counting lemma used in the previous proof. The quantum bounds are due to \cite{NayakWu_Counting_1999} and \cite{Brassard_AmplitudeEstimation_2000}.

\begin{lem}[Approximate counting]\label{lem:approxcount}
    Let $x\in\{0,1\}^n$ be a string to which we have standard classical or quantum query access (i.e., we can query the $i$th bit and receive $x_i$). Then deciding whether $|x|\leq k$ or $|x|\geq k(1+\eps)$ for $1\leq k<n/2$ and $\epsilon \geq 1/k$, requires $\Theta(\min\{\frac{n}{\eps^2k},n\})$ classical queries or $\Theta(\frac{1}{\eps}\sqrt{\frac{n}{k}})$ quantum queries.
\end{lem}

We can now extend the lower bound to larger $S$ and $A$. Before doing so, we will need some structural theorems about quantum query complexity and randomized query complexity. For a function $f$, let $R(f)$ and $Q(f)$ denote their randomized and quantum query complexities.
The first result shows that computing the logical \Or\ of $k$ copies of a problem scales with $k$. The classical result is due to \cite{Goos_GJPW_2017} and the quantum result follows from a general composition theorem for quantum query complexity in \cite{Reichardt_Composition_2011}. The second result, known as a direct sum result, can also be found in \cite{Reichardt_Composition_2011}.

\begin{lem}\label{lem:structural}
    Let $\Or_k$ be the logical $\Or$ function on $k$ bits and $f$ be an arbitrary Boolean function. Then the complexity of the composed function $\Or_k \circ f$, which is defined as the logical \Or\ of the $k$ outputs of $k$ independent instances of $f$ is related to the complexity of $f$ as follows: $Q(\Or_k \circ f) = \Omega(\sqrt{k} \, Q(f))$ and $R(\Or_k \circ f) = \Omega(k R(f))$.
    In addition, computing all $k$ outputs of $k$ independent instances of $f$ requires $\Omega(k\, R(f))$ queries classically and $\Omega(k\, Q(f))$ queries quantumly.
\end{lem}
Note that the ``in addition'' result can be viewed as a result about the query complexity of $f$ composed with the function $\Id_k:\{0,1\}^k\rightarrow\{0,1\}^k; \, x\mapsto x$.

We are now ready to prove the main lower bound theorem.

\begin{proof}[Proof of \cref{thm:lower_bound}]
We start by keeping $S=2$ and allowing arbitrarily large $A\geq 2$. For notational convenience, we identify $\As$ with $\{1,\dots,A\}$.

We will use essentially the same instance as in \cref{fig:lowerbound} but now with $A$ outgoing actions from state $s$, each with transition probability $p_a$ for $a\in \As$. The modified instance is illustrated in \cref{fig:lowerbound_multi}. We again consider the case where all the $p_a$ satisfy the promise that they are either small ($\leq p_0$) or large ($\geq p_0 + \alpha$). As argued in the previous proof, deciding if a given $p_a$ is small or large has a classical lower bound of $\Omega(\Gamma^3/\eps^2)$ and a quantum lower bound of $\Omega(\Gamma^{1.5}/\eps)$.

\begin{figure}[htb]
    \centering
    \includegraphics[scale=0.55]{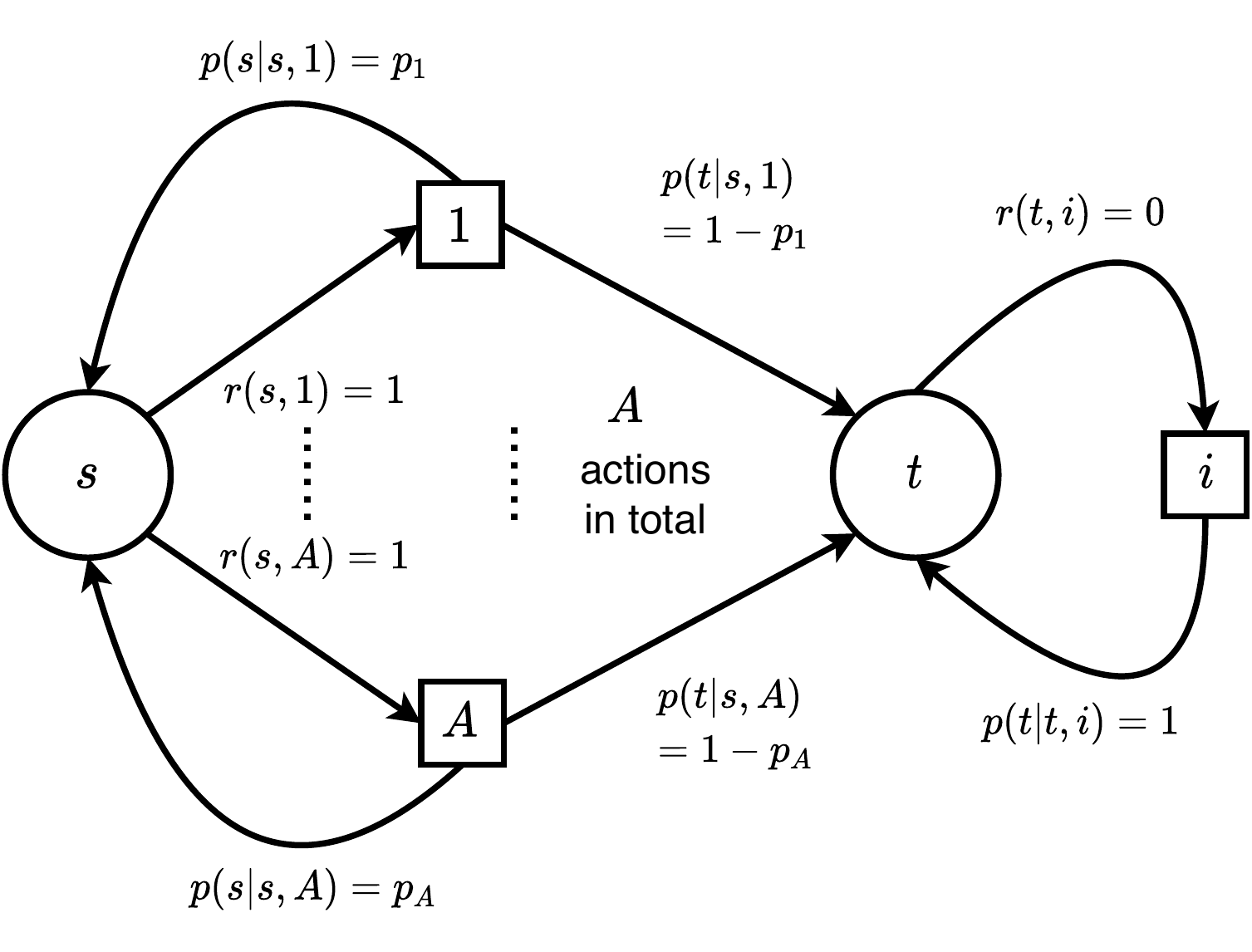}
    \caption{\label{fig:lowerbound_multi}The MDP we use for the lower bound with $S=2$ and arbitrary $A$. For each $i$, $p_i$ is promised to be either $\leq p_0$ or $\geq p_0+\alpha$. Any action $i\in \As$ taken from state $t$ always returns to $t$ with zero reward.}
\end{figure}

Now consider the problem of deciding whether any of the $p_a$ is small or large. This is the logical \Or\ of $A$ independent problems, each of which we have already shown a lower bound for. If we could compute $v^*$ to error $\epsilon$, then we would be able to solve this problem. Hence using \lemref{structural}, we get a classical lower bound of $\Omega(A\Gamma^3/\eps^2)$ and a quantum lower bound of $\Omega(\sqrt{A}\Gamma^{1.5}/\eps)$ for the problem of computing $v^*$. 

Similarly, consider the problem of deciding which of the $p_a$ is large, promised that exactly one of them is large and the rest are small. This is similar to logical \Or, except the goal is to identify the location of a $1$ promised that it exists. This problem is as hard as logical \Or, and we get the same lower bounds. For such an instance, computing $\pi^*$ to error $\eps$ will allow us to distinguish the two cases, since $\pi^*(s)$ should equal the unique action for which $p_a$ is large. This gives us the claimed lower bounds for $\pi^*$.

Similarly, consider the problem of learning which $p_a$s are large and which are small for all $a$ (without any promise on the number of each type). This is the problem of solving $A$ independent instances of a problem for which we have already proved a lower bound. For quantum and classical algorithms, this increases the complexity by a factor of $A$ as stated in the second part of \lemref{structural}. Thus we get a classical lower bound of $\Omega(A\Gamma^3/\eps^2)$ and a quantum lower bound of $\Omega(A\Gamma^{1.5}/\eps)$ for this problem. But if we could compute $q^*$ to error $\epsilon$, then we would be able to solve this problem since such an estimate encodes whether each $p_a$ is large or small. This gives us the claimed lower bounds for $q^*$.

Thus we have established all the lower bounds for $S=2$ and arbitrary $A$. Finally, to extend the lower bounds to arbitrarily large $S$, we can just use $S/2$ copies of the MDP in \cref{fig:lowerbound_multi}. Computing any one of the quantities $q^*$, $v^*$, or $\pi^*$ on this MDP instance means solving $S/2$ independent copies of the problems discussed above. As stated in the second part of \lemref{structural}, for both classical and quantum algorithms, this increases the complexity by a factor of $\Omega(S)$. This yields the claimed lower bounds for general $S$ and $A$.
\end{proof}

\section{Conclusion}

To the best of our knowledge, ours is the first work to rigorously study quantum algorithms for solving MDPs. We show that quantum computers can offer quadratic speedups in terms of $\Gamma$, $\epsilon$, and $A$ in calculating $q^*$, $v^*$, and $\pi^*$. We show our algorithms are either optimal, or optimal assuming $\Gamma$ or $A$ is constant, for certain ranges of $\epsilon$. We discuss some open problems left from our work:

\begin{enumerate}[topsep=0pt,itemsep=0ex,partopsep=1ex,parsep=1ex, leftmargin=*]
\item Can we give optimal algorithms in all parameters ($S$,\, $A$,\, $\Gamma$,\, $\epsilon$) for an unrestricted range of $\epsilon$? A first step towards answering this question may be to try to interpolate between \SolveMdpI\ and \SolveMdpII\ by adjusting the number of epochs and the length of each epoch. This question partly reduces to the purely classical question of finding a sample-optimal algorithm for $v^*$ and $\pi^*$ that has \textit{space} complexity $\Theta(S)$ instead of $\Theta(SA)$.

\item Can we circumvent our quantum lower bounds? In our work, we made few assumptions on the MDP. For special classes of MDPs, there may be greater quantum speedups that break our current quantum lower bounds. Such speedups may also be available in the function approximation setting or if we only ask for a few entries of the vectors $q^*$, $v^*$, and $\pi^*$. For example, see \cite{Ambainis_Dp_2019}.

\item Can we quantize model-based classical algorithms?  Our quantum algorithms are all model-free. But classically, the current best MDP solver is model-based~\cite{Li_TightUpper_2020}. Therefore it is natural to try to construct a quantum model-based algorithm. 
\end{enumerate}

%=============================================================================
\section*{Acknowledgements}
We especially thank Wen Sun for suggesting the tabular MDP setting as the first place to search for quantum speedups and for referring us to \cite{AgarwalJiangKakadeSun_Reinforcement_2021}. We also thank Aaron Sidford, Mengdi Wang, and Xian Wu for helpful discussions on \cite{Sidford_SWWY_2021}. DW acknowledges funding by the Army Research Office (grant W911NF-20-1-0015) and NSF award DMR-1747426. Part of this work was performed while DW was an intern at Microsoft.

%=============================================================================

\appendix

\section{Construction of the Quantum Generative Model}\label{sec:app_oracle_construction}
In this appendix, we describe how to systematically and efficiently construct the quantum generative model (\cref{def:generative_oracle}) from a circuit $\C$ for a classical generative model.

Recall the definition of a classical generative model: for a given state-action pair $(s,a)\in (\Ss,\As)$, $\C$ generates $s'$ with probability $p(s'|s,a)$. Since $\C$ is the circuit of a randomized algorithm, it can be represented as a \emph{deterministic} circuit that takes in two inputs $(s,a)\in (\Ss,\As)$ and $x\in \{0,1\}^m$, and outputs $s'\in \Ss$ with
\begin{equation}
    \pr_{x \sim_U\{0,1\}^m}(\C(s,a,x)=s') = p(s'|s,a),
\end{equation}
where $x \sim_U\{0,1\}^m$ means $x$ is uniformly selected from $\{0,1\}^m$. That is
\begin{equation}\label{eq:count_x}
    | \{x\in \{0,1\}^m \mid  \C(s,a,x)= s' \} | = 2^m\cdot p(s'|s,a).
\end{equation}
Pictorially, $\C$ is of the form:
\begin{equation}
\begin{aligned}
\Qcircuit @C=1em @R=.7em {\lstick{\Ss\times \As \ni  (s,a)} & \qw & \multigate{1}{\quad \C \quad} & \qw & \hspace{10mm}\C(s,a,x) \\
\lstick{\{0,1\}^m \ni x} & \qw & \ghost{\quad \C \quad} &  & \\
}
\end{aligned}
\end{equation}

Now, as $\C$ is a deterministic circuit, we can systematically make it a reversible classical circuit by \cite{Ben73} (see \cite[Sec.~1.4.1]{NielsenChuang_QuantumComputation_2000} for a textbook exposition). This gives another circuit, $\C'$, consisting of $O(\text{size}(\C))$ Toffoli and $\notgate$ gates that uses an additional $O(\text{size}(\C))$ ancillary bits\footnote{We can be more precise if $\C$ consists entirely of $\nandgate$ gates ($\nandgate$ is universal for classical computation). In this case, $\C'$ would, at most, consist of $2\times \text{size}(\C)$ Toffoli gates and $2\times \text{size}(\C)$ $\notgate$ gates, and use an additional $\text{size}(\C)$ ancillary bits.}, of the form:

\begin{equation}
\begin{aligned}
\Qcircuit @C=1em @R=.7em {\lstick{\Ss\times \As \ni  (s,a)} & \qw & \multigate{3}{\quad \C' \quad} & \qw & \hspace{10mm} (s,a) \\
\lstick{\{0,1\}^m \ni x} & \qw & \ghost{\quad \C' \quad} & \qw & \hspace{10mm} x \\
\lstick{\Ss \ni 0_{\Ss}} & \qw & \ghost{\quad \C' \quad} & \qw & \hspace{10mm} \C(s,a,x) \\
\lstick{0^n} & \qw & \ghost{\quad \C' \quad} & \qw & \hspace{10mm} 0^{n}
}
\end{aligned}
\end{equation}
where the $0_{\Ss}$ input of the third wire represents some fixed state in $\Ss$, the wires at the same height (to the left and right of $\C'$) use registers of the same size, and $0^n$ is an ancillary bitstring with $n = O(\text{size}(\C))$.

Now, we can change all the classical Toffoli and $\notgate$ gates in $\C'$ into quantum Toffoli and Pauli-$X$ gates (note that this changes the physical implementation of $\C'$) to produce a quantum circuit $\Q'$. $\Q'$ behaves the same as $\C'$ on classical inputs but is now also able to accept quantum superpositions of these classical inputs. Pictorially, $\Q'$ is of the form:
\begin{equation}\label{eq:qprime_circuit}
\begin{aligned}
\Qcircuit @C=1em @R=.7em {\lstick{\mathbb{C}^{S\times A}\ni\ket{s,a}} & \qw & \multigate{3}{\quad \Q' \quad} & \qw & \hspace{10mm} \ket{s,a} \\
\lstick{(\mathbb{C}^{2})^{\otimes m} \ni\ket{x}} & \qw & \ghost{\quad \Q' \quad} & \qw & \hspace{10mm} \ket{x} \\
\lstick{\mathbb{C}^S \ni \ket{0_{\mathcal{S}}}} & \qw & \ghost{\quad \Q' \quad} & \qw & \hspace{10mm} \ket{\C(s,a,x)} \\
\lstick{\ket{0}^{\otimes n}} & \qw & \ghost{\quad \Q' \quad} & \qw & \hspace{10mm} \ket{0}^{\otimes n}
}
\end{aligned}
\end{equation}

Now, we append a Hadamard gate to each of the $m$ qubits of the second register of $\Q'$ at the start of the computation to give $\Q$. Pictorially, $\Q$ is of the form:
\begin{equation}
\begin{aligned}
\Qcircuit @C=1em @R=.7em 
{
&  & \Q & & &  
\\
&  &  & & &  
\\
\lstick{\mathbb{C}^{S\times A}\ni\ket{s,a}} & \qw & \qw & \multigate{3}{\quad \Q' \quad} & \qw & \hspace{10mm}
\\
\lstick{(\mathbb{C}^{2})^{\otimes m} \ni\ket{0}^{\otimes m}} &\qw & \gate{H^{\otimes m}} & \ghost{\quad \Q' \quad} & \qw & \hspace{10mm}
\\
\lstick{\mathbb{C}^S \ni \ket{0_{\Ss}}} & \qw & \qw & \ghost{\quad \Q' \quad} & \qw &
\\
\lstick{\ket{0}^{\otimes n}} & \qw & \qw & \ghost{\quad \Q' \quad} & \qw & \hspace{10mm} \gategroup{3}{2}{6}{4}{.7em}{--}
}
\end{aligned}
\end{equation}
We can compute the output of $\Q$ on the input $\ket{s,a}\ket{0}^{\otimes m}\ket{0_{\Ss}}\ket{0}^{\otimes n} \equiv \ket{s,a,0^m,0_{\Ss},0^n}$ as follows.

\begin{align*}
    \Q \ket{s,a,0^m,0_{\Ss},0^n} &= \Q' \frac{1}{\sqrt{2^m}}\sum_{x\in \{0,1\}^m} \ket{s,a,x,0_{\Ss},0^n} &&\text{(apply $H^{\otimes m}$)}
    \\
    &= \frac{1}{\sqrt{2^m}}\sum_{x\in \{0,1\}^m} \ket{s,a,x,\C(s,a,x),0^n} &&\text{(\cref{eq:qprime_circuit})}
    \\
    &= \ket{s,a}\sum_{s' \in \Ss} \Bigl( \frac{1}{\sqrt{2^m}}\sum_{x\in \{0,1\}^m \, \mid \, \C(s,a,x)=s'}\ket{x} \Bigr) \ket{s'}\ket{0^n} &&\text{(rearrange sum).}
\end{align*}

Now, the $m$-qubit state in the brackets has norm $|\{x\in \{0,1\}^m \, \mid \, \C(s,a,x)=s'\}|/2^m = p(s'|s,a)$ due to \cref{eq:count_x}. Therefore, it can be rewritten as $\sqrt{p(s'|s,a)}\ket{\psi_{s,a,s'}}$ for some normalized state $\ket{\psi_{s,a,s'}}$. Therefore
\begin{equation}
\Q \ket{s,a,0^m,0_{\Ss},0^n} = \ket{s,a}\sum_{s' \in \Ss}  \sqrt{p(s'|s,a)}\ket{\psi_{s,a,s'}}\ket{s'} \ket{0^n}.
\end{equation}

Finally, swapping the $\ket{\psi_{s,a,s'}}$ and $\ket{s'}$ registers,  dropping the $\ket{0^n}$ register for convenience, and renaming $0_{\Ss}$ to just $0$, we see that $\Q$ is precisely of the form in \cref{def:generative_oracle}.

\bibliographystyle{alphaurl}
\bibliography{references}

\newcommand{\etalchar}[1]{$^{#1}$}
\newcommand{\arxiv}[1]{
  \href{https://arxiv.org/abs/#1}{\ttfamily{arXiv:#1}}\?}\newcommand{\quantum}{Quantum}\newcommand{\prl}{Physical
  Review Letters}\newcommand{\nature}{Nature}\newcommand{\tqc}[1]{#1 Conference
  on the Thoery of Quantum Computation, Communication, and Cryptography
  (TQC)}\newcommand{\jmlr}{Journal of Machine Learning
  Research}\newcommand{\neurips}[1]{Advances in Neural Information Processing
  Systems #1 ({NeurIPS})}\newcommand{\icml}[1]{Proceedings of the #1
  International Conference on Machine Learning (ICML)}\newcommand{\toc}{Theory
  of Computing}\newcommand{\icalp}[1]{#1 International Colloquium on Automata,
  Languages, and Programming ({ICALP})}\newcommand{\stoc}[1]{Proceedings of the
  #1 {ACM} Symposium on the Theory of Computing
  ({STOC})}\newcommand{\soda}[1]{Proceedings of the #1 {ACM-SIAM} Symposium on
  Discrete Algorithms ({SODA})}\newcommand{\colt}[1]{Proceedings of the #1
  Conference On Learning Theory (COLT)}\newcommand{\qic}{Quantum Information
  and Computation}\newcommand{\prx}{Physical Review
  X}\def\?#1{\if.#1{}\else#1\fi}
\begin{thebibliography}{JTPN{\etalchar{+}}21}

\bibitem[AAB{\etalchar{+}}19]{Google_QuantumSupremacy_2019}
Frank Arute, Kunal Arya, Ryan Babbush, Dave Bacon, Joseph~C. Bardin, Rami
  Barends, Rupak Biswas, Sergio Boixo, Fernando G. S.~L. Brandao, David~A.
  Buell, Brian Burkett, Yu~Chen, Zijun Chen, Ben Chiaro, Roberto Collins,
  William Courtney, Andrew Dunsworth, Edward Farhi, \ldots, Hartmut Neven, and
  John~M. Martinis.
\newblock {Quantum supremacy using a programmable superconducting processor}.
\newblock {\em Nature}, 574(7779):505--510, 2019.
\newblock \href {http://dx.doi.org/10.1038/s41586-019-1666-5}
  {\path{doi:10.1038/s41586-019-1666-5}}.

\bibitem[ABI{\etalchar{+}}19]{Ambainis_Dp_2019}
Andris Ambainis, Kaspars Balodis, J{\=a}nis Iraids, Martins Kokainis, Kri{\v
  s}j{\=a}nis Pr{\=u}sis, and Jevg{\=e}nijs Vihrovs.
\newblock Quantum speedups for exponential-time dynamic programming algorithms.
\newblock In {\em \soda{30th}}, pages 1783--1793. Society for Industrial and
  Applied Mathematics, 2019.
\newblock \href {http://dx.doi.org/10.1137/1.9781611975482.107}
  {\path{doi:10.1137/1.9781611975482.107}}.

\bibitem[AJKW21]{AgarwalJiangKakadeSun_Reinforcement_2021}
Alekh Agarwal, Nan Jiang, Sham~M. Kakade, and Sun Wen.
\newblock {Reinforcement Learning: Theory and Algorithms}.
\newblock Book in preparation, draft available at
  \href{https://rltheorybook.github.io/}{rltheorybook.github.io}. Date
  accessed: 2nd November, 2021.

\bibitem[AKY20]{Agarwal_MinimaxOptimal_2020}
Alekh Agarwal, Sham Kakade, and Lin~F. Yang.
\newblock {Model-Based Reinforcement Learning with a Generative Model is
  Minimax Optimal}.
\newblock In {\em \colt{33rd}}, volume 125 of {\em Proceedings of Machine
  Learning Research}, pages 67--83. PMLR, 2020.

\bibitem[AMK12]{AzarMunosKappen_MdpGenerative_2012}
Mohammad~Gheshlaghi Azar, R\'{e}mi Munos, and Hilbert~J. Kappen.
\newblock {On the Sample Complexity of Reinforcement Learning with a Generative
  Model}.
\newblock In {\em \icml{29th}}, pages 1707--1714, 2012.

\bibitem[Bel19]{Belovs_Quantum_2019}
Aleksandrs Belovs.
\newblock {Quantum Algorithms for Classical Probability Distributions}, 2019.
\newblock \arxiv{1904.02192}.

\bibitem[Ben73]{Ben73}
C.~H. Bennett.
\newblock {Logical Reversibility of Computation}.
\newblock {\em IBM J. Res. Dev.}, 17(6):525–532, 1973.
\newblock \href {http://dx.doi.org/10.1147/rd.176.0525}
  {\path{doi:10.1147/rd.176.0525}}.

\bibitem[Ber00]{Berteskas_DynamicControl_2000}
D.P. Bertsekas.
\newblock {\em {Dynamic Programming and Optimal Control}}.
\newblock Number v. 1 in Athena Scientific optimization and computation series.
  Athena Scientific, 2000.

\bibitem[Ber13]{Bertsekas_Abstract_2013}
D.P. Bertsekas.
\newblock {\em {Abstract Dynamic Programming}}.
\newblock Athena Scientific, 2013.

\bibitem[BHMT02]{Brassard_AmplitudeEstimation_2000}
Gilles Brassard, Peter Høyer, Michele Mosca, and Alain Tapp.
\newblock Quantum amplitude amplification and estimation.
\newblock {\em Quantum Computation and Information}, page 53–74, 2002.
\newblock \href {http://dx.doi.org/10.1090/conm/305/05215}
  {\path{doi:10.1090/conm/305/05215}}.

\bibitem[CJ21]{Cornelissen_MultiMonte_2021}
Arjan Cornelissen and Sofiene Jerbi.
\newblock {Quantum algorithms for multivariate Monte Carlo estimation}, 2021.
\newblock \arxiv{2107.03410}.

\bibitem[DCLT08]{Dong_QuantumRl_2008}
D.~Dong, C.~Chen, H.~Li, and T.~Tarn.
\newblock {Quantum Reinforcement Learning}.
\newblock {\em IEEE Transactions on Systems, Man, and Cybernetics, Part B
  (Cybernetics)}, 38(5):1207--1220, 2008.
\newblock \href {http://dx.doi.org/10.1109/TSMCB.2008.925743}
  {\path{doi:10.1109/TSMCB.2008.925743}}.

\bibitem[DH96]{Durr_MinFinding_1996}
Christoph D{\"u}rr and Peter H{\o}yer.
\newblock {A Quantum Algorithm for Finding the Minimum}.
\newblock 1996.
\newblock \arxiv{quant-ph/9607014}.

\bibitem[DTB16]{Dunjko_MachineLearning_2016}
Vedran Dunjko, Jacob~M. Taylor, and Hans~J. Briegel.
\newblock {Quantum-Enhanced Machine Learning}.
\newblock {\em \prl}, 117(13), 2016.
\newblock \href {http://dx.doi.org/10.1103/PhysRevLett.117.130501}
  {\path{doi:10.1103/PhysRevLett.117.130501}}.

\bibitem[DTB17]{Dunjko_Advances_2017}
V.~{Dunjko}, J.~M. {Taylor}, and H.~J. {Briegel}.
\newblock {Advances in quantum reinforcement learning}.
\newblock In {\em 2017 IEEE International Conference on Systems, Man, and
  Cybernetics (SMC)}, pages 282--287, 2017.
\newblock \href {http://dx.doi.org/10.1109/SMC.2017.8122616}
  {\path{doi:10.1109/SMC.2017.8122616}}.

\bibitem[GJPW17]{Goos_GJPW_2017}
Mika G{\"o}{\"o}s, T.~S. Jayram, Toniann Pitassi, and Thomas Watson.
\newblock {Randomized Communication vs. Partition Number}.
\newblock In {\em \icalp{44th}}, volume~80 of {\em Leibniz International
  Proceedings in Informatics (LIPIcs)}, pages 52:1--52:15. Schloss
  Dagstuhl--Leibniz-Zentrum fuer Informatik, 2017.
\newblock \href {http://dx.doi.org/10.4230/LIPIcs.ICALP.2017.52}
  {\path{doi:10.4230/LIPIcs.ICALP.2017.52}}.

\bibitem[GLM08]{GiovannettiLloydMaccone_Qram_2008}
Vittorio Giovannetti, Seth Lloyd, and Lorenzo Maccone.
\newblock {Quantum Random Access Memory}.
\newblock {\em \prl}, 100:160501, 2008.
\newblock \href {http://dx.doi.org/10.1103/PhysRevLett.100.160501}
  {\path{doi:10.1103/PhysRevLett.100.160501}}.

\bibitem[Gro96]{Grover_Search_1996}
Lov~K. Grover.
\newblock {A fast quantum mechanical algorithm for database search}.
\newblock In {\em \stoc{28th}}, pages 212--219, New York, NY, USA, 1996.
  Association for Computing Machinery.
\newblock \href {http://dx.doi.org/10.1145/237814.237866}
  {\path{doi:10.1145/237814.237866}}.

\bibitem[Ham21]{Hamoudi_SubGaussian_2021}
Yassine Hamoudi.
\newblock {Quantum Sub-Gaussian Mean Estimator}.
\newblock In {\em 29th Annual European Symposium on Algorithms (ESA 2021)},
  volume 204 of {\em Leibniz International Proceedings in Informatics
  (LIPIcs)}, pages 50:1--50:17, Dagstuhl, Germany, 2021. Schloss Dagstuhl --
  Leibniz-Zentrum f{\"u}r Informatik.
\newblock \href {http://dx.doi.org/10.4230/LIPIcs.ESA.2021.50}
  {\path{doi:10.4230/LIPIcs.ESA.2021.50}}.

\bibitem[HM19]{Hamoudi_Chebyshev_2019}
Yassine Hamoudi and Fr{\'e}d{\'e}ric Magniez.
\newblock {Quantum Chebyshev's Inequality and Applications}.
\newblock In {\em \icalp{46th}}, volume 132 of {\em Leibniz International
  Proceedings in Informatics (LIPIcs)}, pages 69:1--69:16. Schloss
  Dagstuhl--Leibniz-Zentrum fuer Informatik, 2019.
\newblock \href {http://dx.doi.org/10.4230/LIPIcs.ICALP.2019.69}
  {\path{doi:10.4230/LIPIcs.ICALP.2019.69}}.

\bibitem[JS20]{Sidford_Mirror_2020}
Yujia Jin and Aaron Sidford.
\newblock {Efficiently Solving {MDP}s with Stochastic Mirror Descent}.
\newblock In {\em \icml{37th}}, volume 119 of {\em Proceedings of Machine
  Learning Research}, pages 4890--4900. PMLR, 2020.

\bibitem[JTPN{\etalchar{+}}21]{Jerbi_Rl_2020}
Sofiene Jerbi, Lea~M. Trenkwalder, Hendrik Poulsen~Nautrup, Hans~J. Briegel,
  and Vedran Dunjko.
\newblock {Quantum Enhancements for Deep Reinforcement Learning in Large
  Spaces}.
\newblock {\em PRX Quantum}, 2:010328, 2021.
\newblock \href {http://dx.doi.org/10.1103/PRXQuantum.2.010328}
  {\path{doi:10.1103/PRXQuantum.2.010328}}.

\bibitem[JVV86]{Jerrum_Powering_1986}
Mark~R. Jerrum, Leslie~G. Valiant, and Vijay~V. Vazirani.
\newblock {Random generation of combinatorial structures from a uniform
  distribution}.
\newblock {\em Theoretical Computer Science}, 43:169--188, 1986.
\newblock \href {http://dx.doi.org/10.1016/0304-3975(86)90174-X}
  {\path{doi:10.1016/0304-3975(86)90174-X}}.

\bibitem[Kak03]{Kakade_Thesis_2003}
Sham~M. Kakade.
\newblock {\em {On the Sample Complexity of Reinforcement Learning}}.
\newblock PhD thesis, 2003.

\bibitem[KMN02]{KearnsMansourNg_SparseSampling_2002}
Michael Kearns, Yishay Mansour, and Andrew~Y. Ng.
\newblock {A Sparse Sampling Algorithm for Near-Optimal Planning in Large
  Markov Decision Processes}.
\newblock {\em Machine Learning}, 49(2):193--208, 2002.
\newblock \href {http://dx.doi.org/10.1023/A:1017932429737}
  {\path{doi:10.1023/A:1017932429737}}.

\bibitem[KS99]{KearnsSingh_PhasedQlearning_1999}
Michael Kearns and Satinder Singh.
\newblock {Finite-Sample Convergence Rates for Q-Learning and Indirect
  Algorithms}.
\newblock In {\em \neurips{11}}, pages 996--1002, Cambridge, MA, USA, 1999. MIT
  Press.

\bibitem[LWC{\etalchar{+}}20]{Li_TightUpper_2020}
Gen Li, Yuting Wei, Yuejie Chi, Yuantao Gu, and Yuxin Chen.
\newblock {Breaking the Sample Size Barrier in Model-Based Reinforcement
  Learning with a Generative Model}.
\newblock In {\em \neurips{33}}, volume~33, pages 12861--12872. Curran
  Associates, Inc., 2020.

\bibitem[Mon15]{Montanaro_MonteCarlo_2015}
Ashley Montanaro.
\newblock {Quantum speedup of Monte Carlo methods}.
\newblock {\em Proceedings of the Royal Society A: Mathematical, Physical and
  Engineering Sciences}, 471(2181):20150301, 2015.
\newblock \href {http://dx.doi.org/10.1098/rspa.2015.0301}
  {\path{doi:10.1098/rspa.2015.0301}}.

\bibitem[NC00]{NielsenChuang_QuantumComputation_2000}
Michael~A. Nielsen and Isaac~L. Chuang.
\newblock {\em {Quantum Computation and Quantum Information}}.
\newblock Cambridge University Press, 2000.

\bibitem[NW99]{NayakWu_Counting_1999}
Ashwin Nayak and Felix Wu.
\newblock {The quantum query complexity of approximating the median and related
  statistics}.
\newblock In {\em \stoc{31st}}, pages 384--393, New York, NY, USA, 1999.
  Association for Computing Machinery.
\newblock \href {http://dx.doi.org/10.1145/301250.301349}
  {\path{doi:10.1145/301250.301349}}.

\bibitem[PDM{\etalchar{+}}14]{Paparo_ActiveLearning_2014}
Giuseppe~Davide Paparo, Vedran Dunjko, Adi Makmal, Miguel~Angel Martin-Delgado,
  and Hans~J. Briegel.
\newblock {Quantum Speedup for Active Learning Agents}.
\newblock {\em \prx}, 4(3):031002, 2014.
\newblock \href {http://dx.doi.org/10.1103/PhysRevX.4.031002}
  {\path{doi:10.1103/PhysRevX.4.031002}}.

\bibitem[Rei11]{Reichardt_Composition_2011}
Ben~W. Reichardt.
\newblock {Reflections for Quantum Query Algorithms}.
\newblock In {\em \soda{22nd}}, pages 560--569, USA, 2011. Society for
  Industrial and Applied Mathematics.

\bibitem[SB18]{SB_ReinforcementLearning_2018}
Richard~S. Sutton and Andrew~G. Barto.
\newblock {\em {Reinforcement Learning: An Introduction}}.
\newblock The MIT Press, Cambridge, MA, USA, 2018.

\bibitem[Sho97]{Shor_Factoring_1997}
Peter~W. Shor.
\newblock {Polynomial-Time Algorithms for Prime Factorization and Discrete
  Logarithms on a Quantum Computer}.
\newblock {\em SIAM Journal on Computing}, 26(5):1484--1509, 1997.
\newblock \href {http://dx.doi.org/10.1137/S0097539795293172}
  {\path{doi:10.1137/S0097539795293172}}.

\bibitem[SWW{\etalchar{+}}18]{Sidford_NearOptimal_2018}
Aaron Sidford, Mengdi Wang, Xian Wu, Lin Yang, and Yinyu Ye.
\newblock {Near-Optimal Time and Sample Complexities for Solving Markov
  Decision Processes with a Generative Model}.
\newblock In {\em \neurips{31}}, pages 5186--5196. Curran Associates, Inc.,
  2018.

\bibitem[SWWY21]{Sidford_SWWY_2021}
Aaron Sidford, Mengdi Wang, Xian Wu, and Yinyu Ye.
\newblock {Variance reduced value iteration and faster algorithms for solving
  Markov decision processes}.
\newblock {\em Naval Research Logistics (NRL)}, 2021.
\newblock \href {http://dx.doi.org/10.1002/nav.21992}
  {\path{doi:10.1002/nav.21992}}.

\bibitem[SY94]{SinghYee_PolicyLoss_1994}
Satinder~P. Singh and Richard~C. Yee.
\newblock {An upper bound on the loss from approximate optimal-value
  functions}.
\newblock {\em Machine Learning}, 16(3):227--233, 1994.
\newblock \href {http://dx.doi.org/10.1007/BF00993308}
  {\path{doi:10.1007/BF00993308}}.

\bibitem[Sze10]{Szepesvari_AlgorithmsRl_2010}
Csaba Szepesv{\'a}ri.
\newblock {Algorithms for Reinforcement Learning}.
\newblock {\em Synthesis Lectures on Artificial Intelligence and Machine
  Learning}, 4(1):1--103, 2010.
\newblock \href {http://dx.doi.org/10.2200/S00268ED1V01Y201005AIM009}
  {\path{doi:10.2200/S00268ED1V01Y201005AIM009}}.

\bibitem[vA21]{VanApeldoorn_MultiDimension_2021}
Joran van Apeldoorn.
\newblock {Quantum Probability Oracles \& Multidimensional Amplitude
  Estimation}.
\newblock In {\em \tqc{16th}}, volume 197 of {\em Leibniz International
  Proceedings in Informatics (LIPIcs)}, pages 9:1--9:11. Schloss Dagstuhl --
  Leibniz-Zentrum f{\"u}r Informatik, 2021.
\newblock \href {http://dx.doi.org/10.4230/LIPIcs.TQC.2021.9}
  {\path{doi:10.4230/LIPIcs.TQC.2021.9}}.

\bibitem[Wai19]{Wainwright_VarianceReduced_2019}
Martin~J. Wainwright.
\newblock {Variance-reduced $Q$-learning is minimax optimal}, 2019.
\newblock \arxiv{1906.04697}.

\bibitem[Wan17]{Wang_PrimalDual_2017}
Mengdi Wang.
\newblock {Primal-Dual $\pi$ Learning: Sample Complexity and Sublinear Run Time
  for Ergodic Markov Decision Problems}, 2017.
\newblock \arxiv{1710.06100}.

\bibitem[Wan20]{Wang_Randomized_2020}
Mengdi Wang.
\newblock {Randomized Linear Programming Solves the Markov Decision Problem in
  Nearly Linear (Sometimes Sublinear) Time}.
\newblock {\em Mathematics of Operations Research}, 45(2):517--546, 2020.
\newblock \href {http://dx.doi.org/10.1287/moor.2019.1000}
  {\path{doi:10.1287/moor.2019.1000}}.

\bibitem[WSK{\etalchar{+}}21]{ConferenceVersion_Rl_2021}
Daochen Wang, Aarthi Sundaram, Robin Kothari, Ashish Kapoor, and Martin
  Roetteler.
\newblock Quantum algorithms for reinforcement learning with a generative
  model.
\newblock In {\em Proceedings of the 38th International Conference on Machine
  Learning}, volume 139 of {\em Proceedings of Machine Learning Research},
  pages 10916--10926. PMLR, 18--24 Jul 2021.

\bibitem[WYLC21]{WangYouLiChilds_Bandits_2021}
Daochen Wang, Xuchen You, Tongyang Li, and Andrew~M. Childs.
\newblock {Quantum Exploration Algorithms for Multi-Armed Bandits}.
\newblock {\em Proceedings of the AAAI Conference on Artificial Intelligence},
  35(11):10102--10110, 2021.

\end{thebibliography}
\end{document}